\documentclass[12pt]{article}

\pdfoutput=1

\usepackage{upgreek}
\usepackage[greek,english]{babel}
\usepackage{amsmath}
\usepackage{amssymb}
\usepackage{amsfonts}
\usepackage{amsthm}
\usepackage{graphicx}
\usepackage{wrapfig}
\usepackage{mathtools}
\usepackage{txfonts}
\usepackage{mathrsfs}
\usepackage{longtable}

\newcommand{\intl}[2]{\!\underset{#1}{\overset{#2}{\rotatebox[origin=rc]{15}{\large\ensuremath{\int}}}}}
\newcommand{\intL}[2]{\!\!\underset{#1}{\overset{#2}{\rotatebox[origin=rc]{15}{\Large\ensuremath{\int}}}}\!}

\newcommand{\Conv}{\mathop{\scalebox{2}{\raisebox{-0.2ex}{$\ast$}}}}

\newcommand{\z}{\textsl{z}}

\newcommand{\D}{\partial}

\newcommand{\deq}{\overset{\operatorname{def}}{=}}

\newcommand{\supp}{\operatorname{supp}}
\newcommand{\ch}{\operatorname{ch}}
\newcommand{\diam}{\operatorname{diam}}

\renewcommand{\alpha}{\alphaup}

\renewcommand{\beta}{\betaup}
\renewcommand{\gamma}{\gammaup}
\renewcommand{\delta}{\deltaup}

\renewcommand{\epsilon}{\varepsilonup}
\renewcommand{\varepsilon}{\epsilonup}

\renewcommand{\zeta}{\zetaup}

\renewcommand{\eta}{\etaup}
\renewcommand{\theta}{\thetaup}
\renewcommand{\vartheta}{\varthetaup}

\renewcommand{\iota}{\iotaup}

\renewcommand{\kappa}{\varkappa}
\renewcommand{\lambda}{\lambdaup}

\renewcommand{\mu}{\muup}

\renewcommand{\nu}{\nuup}
\renewcommand{\xi}{\xiup}

\renewcommand{\pi}{\piup}

\renewcommand{\rho}{\rhoup}
\renewcommand{\varrho}{\varrhoup}
\renewcommand{\sigma}{\sigmaup}
\renewcommand{\varsigma}{\varsigmaup}

\renewcommand{\tau}{\tauup}
\renewcommand{\Upsilon}{\textrm{\greektext U}}
\renewcommand{\upsilon}{\upsilonup}
\renewcommand{\phi}{\upvarphi}
\renewcommand{\varphi}{\phiup}

\renewcommand{\chi}{\chiup}
\renewcommand{\psi}{\textrm{\greektext y}}
\renewcommand{\omega}{\omegaup}

\renewcommand{\mathbb}{\varmathbb}

\title{Natural selection in compartmentalized environment with reshuffling}

\author{A.S. Zadorin$^1$ \and Y. Rondelez$^2$}

\date{\footnotesize
	$^1$Gulliver, ESPCI Paris, PSL University, CNRS, 75005 Paris, France\\
	(Current affiliation: Chimie Biologie Innovation, ESPCI Paris, CNRS, PSL University, 75005 Paris, France)\\
	email: \texttt{anton.zadorin@espci.fr}\\
	$^2$Gulliver, ESPCI Paris, PSL University, CNRS, 75005 Paris, France\\
	email: \texttt{yannick.rondelez@espci.fr}
}

\newtheorem{property}{Property}
\newtheorem*{lemma-main}{Lemma}
\newtheorem{theorem}{Theorem}[subsection]
\newtheorem{lemma}[theorem]{Lemma}
\newtheorem{proposition}[theorem]{Proposition}

\begin{document}

\maketitle

\begin{abstract}

\sloppy{
The emerging field of high-throughput compartmentalized \emph{in vitro} evolution is a promising new approach to protein engineering. In these 
experiments, libraries of mutant genotypes are randomly distributed and expressed in compartments---droplets of an emulsion. The selection of desirable 
variants is performed according to the phenotype of each compartment. The random partitioning leads to a fraction of compartments receiving more than 
one genotype making the whole process a lab implementation of the group selection. From a practical point of view (where efficient selection is 
typically sought), it is important to know the impact of the increase in the mean occupancy of compartments on the selection efficiency. We carried out 
a theoretical investigation of this problem in the context of selection dynamics for a simple model with an infinite, non-mutating population that is 
periodically partitioned among an infinite number of identical compartments. We derive here an update equation for any distribution of phenotypes and 
any value of the mean occupancy. Using this result, we demonstrate that, for the linear additive fitness, the best genotype is still selected 
regardless of the mean occupancy. Furthermore, the ``natural'' selection process is remarkably resilient to the presence of multiple genotypes per 
compartments, and slows down approximately inversely proportional to the mean occupancy at high values. We extend out results to more general 
expressions that cover nonadditive and non-linear fitnesses, as well non-Poissonian distribution among compartments. Our conclusions may also apply to 
natural genetic compartmentalized replicators, such as viruses or early \emph{trans}-acting RNA replicators.
}

\end{abstract}

\section*{Introduction}

\emph{In vitro} directed evolution is a laboratory technique that mimics natural evolution and can be used to obtain proteins with new or improved 
properties \cite{Packer2015}. Modern approaches are applied to the search of catalytic properties (i.e. artificial enzymes), for which they can test 
millions or billions of variants in parallel. These approaches use microcompartments, such as water-in-oil droplets in an emulsion, in which variants 
are randomly distributed, to enforce the phenotype-genotype linkage between the gene and its protein products. A selection pressure is then applied, at 
the level of compartments, in order to drive the population of enzymes towards the desired property. Cycles of mutation-selection are usually iterated 
until a satisfying variant is obtained.

The technique has been successfully used to improve existing enzymes (in relation to their direct catalytic characteristics, or to thermal stability, 
resistance to inhibitors, etc.), to change the substrate specificity of an enzyme, to develop a completely new enzymatic activity, or even new 
catalytic pathways (specific examples can be found in \cite{Martinez2013,Wojcik2015,Zeymer2018}). Compartmentalized \emph{in vitro} selection has also 
been used to study fundamental questions of protein physics such as the local shape of the protein fitness landscape on its sequence space with respect 
to a particular selection pressure \cite{Romero2015,Zeymer2018}. Here the gene frequency change after a selection round for millions of mutants at once 
is used as a readout.

A typical high throughput directed evolution experiment starts with a large library of mutated genes of interest. Because enzymatic activity is carried 
by the encoded proteins, these genes require transcription and translation for phenotypic expression. This expression step is performed either by 
microbiological means (\emph{via} recombinant expression of a gene-carrying plasmid in bacteria) or with direct \emph{in vitro} approaches. In all 
cases, the mutant library is distributed among a large number of small compartments (e.g. microdroplets within an emulsion). The selection is carried 
out on each compartment, by evaluating its global -or apparent- phenotype and converting this information into an artificial ``fitness''. This critical 
step can be done in two different ways. The first one uses an external observation and subsequent physical separation of ``good'' and ``bad'' 
compartments, depending on the detected phenotype. In many cases, the phenotypic properties are converted into some fluorescent readout and then the 
compartments are sorted based on their spectral properties. An experiment of this kind, pioneering high-throughput emulsion sorting and \emph{in vitro} 
compartmentalization, was performed in \cite{Tawfik1998}.

\sloppy{The second approach is based on an internal biochemical reaction that autonomously replicates the genotype in relation to some metrics of the 
phenotype (this situation is usually referred to as ``self-selection''). This approach is well suited to the selection of DNA- and RNA-polymerases and 
was pioneered in \cite{Ghadessy2001} with the evolution of Taq polymerase towards higher resistance to PCR inhibitors. In that work, bacteria 
containing a plasmid with Taq gene and the corresponding Taq polymerase were randomly encapsulated in water-in-oil emulsion, with the addition of PCR 
primers targeting the Taq gene variant. After bacterial lysis released both plasmid and polymerase, droplets were submitted to thermal cycles and those 
with better polymerases produced more copies of the genetic variants they contained. At the end of the selection cycle the emulsion is broken and the 
genetic mate- rial is collected. If multiple selections rounds are used, the recovered genes serve as the initial library for the subsequent cycle. An 
idealized scheme of the self-selection process is depicted on Figure~\ref{fig-scheme}.
}

Note that while the two selection processes described above correspond to very different experimental setups, from a modelling point of view, one can 
consider the former approach as a special case of the latter, where the fitness and the phenotype in a droplet are related via a step-function. 

An essential feature of all such experiments is that selection does not act on individuals, but rather on groups of individuals which are randomly 
formed at each generation, during the compartmentalization process. When multiple individuals happen to share the same compartment, the selection 
outcome for a given genotype depends not only on its own identity, but also on identities of the others. This is because, inside each compartment, all 
genetic molecules are copied (or sorted) without distinction, but as a function of the combined phenotypic composition in the compartment. Therefore, 
as the statistics of the genotype distribution in compartments depends on the distribution of genotypes in the population, the selection effectively 
becomes frequency-dependent. In fact, these selection systems represent an extreme case of group selection model, where the reproductive success of an 
individual depends entirely on the combined phenotype of the group it belongs to \cite{Wilson1973}.

In the context of the search of new enzymes by high-throughput \emph{in vitro} evolution the problem of random co-encapsulation is of great practical 
significance. The most efficient selection, in the sense of the selection pressure, is achieved in the situation when any compartment contains no more 
than one genotype. Such situation brings no dependence on frequency. However, unless sophisticated methods are used to enforce single occupancy 
\cite{Edd2008,He2005}, this implies a very high fraction of empty droplets, and thus a loss in throughput. A question naturally arises about the 
effect of allowing multiple genotypes in one droplet. Obviously, the selection pressure will drop, because if a stronger genotype $A$ and a weaker 
genotype $B$ meet in the same droplet, the shared replicative phenotype is weaker than that of two $A$-s and stronger than that of two $B$-s, while the 
number of copies is evenly shared. The precise response in terms of selection dynamics, and the relevant parameters, are however less clear. Lacking a 
general understanding of the process dynamics, researches are bound to empirical approaches to perform \emph{in vitro} evolution experiments 
\cite{Collins2015,Dodevski2015}.

The goal of our work was to establish a general model of group selection suitable for studying \emph{in vitro} evolution with random co-encapsulation. 
To be useful for interpreting experiments, the model must have a form of a dynamical equation that governs the temporal evolution of genotypic and 
phenotypic distributions in the library of mutants. The initial data and parameters must be: the initial genotypic or phenotypic distribution, 
the microscopic (on the level of a single group) rules of genotype interaction and its effect of the group survival and/or reproduction, and the 
statistics of the co-encapsulation. The model must eventually allow to study the effect of random co-encapsulation on the selection efficiency. This 
article takes a first step towards that goal, where we assume no selection inside the groups.

The existing general theoretical works on group selection are typically focused on two different aspects of this phenomenon. The first group of works 
studies the conditions on selection of individually disadvantageous but collectively advantageous traits (like altruism). This direction of research 
has produced a very large body of literature. We refer \cite{Gardner2009} and \cite{Queller1992} and the links therein. The main instrument here is 
Price's covariation formalism \cite{Price1970,Price1972}. However, as Price's identity does not provide a dynamically complete equation, this approach 
is not appropriate for the question studied in this article. Furthermore, the construction of this identity is based on the knowledge of the 
individuals' fitness, which itself requires the development of a theory for compartmentalized selection \cite{vanVeelen2005}. The second traditional 
direction is related to the selection of an altruistic trait, too, but the main goal is to find its dynamics \cite{Smith1964,Wilson1975}. These works 
assume groups of the same size, additive fitness effect with altruism cost and benefits. This approach does produce a dynamical equation for the 
selection process starting from the model of interactions. The key point of such models is the interplay between the inter- and intragroup selection. 
Although such models with no intragroup selection do approach closely to the problem of compartmentalized selection, they carry a number of 
limitations. The group size is fixed, the number of possible phenotypic values is finite, and the emphasis on elementary mathematics complicates the 
treatment of nonlinear genotype/phenotype relations. Finally, there is a recent attempt to build a very general dynamic equation of group selection in 
form of PDE \cite{Simon2013}. The assumptions of the model, however, departure from the context of the compartmentalized experiments, too. They include 
a well defined finite set of phenotypes and a continuous-time dynamics of group restructure, while in the experiments, the phenotypic distributions are 
often continuous and the groups are completely reformed at discrete times.

There is also an overlap of the problem of co-encapsulated group selection with the problem of evolution in patched/structured populations, first 
conceived in \cite{Wright1931}. More specifically, the problem in question is isomorphic to a certain limit of the metapopulation model with 
migrant-pool gene flow \cite{Burger2014,Hamilton2011}. Here the role of abstract groups is played by ecological demes. This direction is currently 
experiencing a renewed theoretical interest with the development of strict mathematical models of exact or approximate stochastic evolution dynamics 
with mutations based on Markov process formalism (for a modern treatment see, for example, \cite{Bitbol2014}). However, typical simplifications of 
the models include few phenotypes (in most cases only two), weak migration, and an essential competition inside demes. Compartmentalized directed 
evolution experiments are, instead, characterized by a very large number of phenotypes, panmixing, and no intracompartment selection. Therefore, a new 
approach is needed to tackle the problem.

As the co-encapsulation of multiple genotypes under selection also arises in a number of biologically relevant scenario (e.g. multiple infection for 
viruses, parasites), it has also been the subject of a number of more specific previous work. Some recent examples include the works 
\cite{Bianconi2013,Fontanari2006,Fontanari2013,Lampert2011,Matsumura2016,Zintzaras2010}, as well as \cite{Higgs2015} and links therein. Most of them, 
however, focus on aspects of kin/group selection, altruistic trait fixation, coexistence of selfish and cooperative genes, and the error catastrophe. 
Typically, these works treat more complicated cases of primordial evolution in presence of parasitic sequences with different reproduction rates. These 
complex problems are explored primarily with numerical simulations, and are not directly applicable to the present problem of \emph{in vitro} 
evolution. A preliminary analysis of the selection dynamics directly related to the context of \emph{in vitro} co-compartmentalization was done in 
\cite{Zheng2007} for the case of two alleles, one of which is completely inactive, where any droplet with at least one active genotype inside is 
selected and all its content propagates to the next generation. A similar study of the effect of co-infection on selection from two viral phenotypes 
that share, when located in the same cell, both the replicative activity and the total offspring number was also carried out in \cite{Novella2004}.

Here we present a more general description of the compartmentalized selection problem, that we initially motivate on a model related to \emph{in 
vitro} evolution experiments in droplets, but which may apply also to biological situations. We consider an infinite population distributed in an 
infinite number of identical compartments, where the reproduction happens at discrete moments of time, the generations do not overlap, and the 
population is randomly redistributed at each cycle. We take into account only selection and ignore mutations and stochasticity. We derive, and solve 
for special cases, the update equation that defines the temporal evolution of the phenotypic distribution. Such approach explicitly incorporates the 
dependence of the selection pressure on the mean occupancy of droplets in the emulsion.

The structure of the article is the following. Section~\ref{model} gives the detailed description of the model. Section~\ref{selection-arbitrary} deals 
with the general theory for an additive phenotype and a linear phenotype-fitness dependence. There we derive the probability distribution densities of 
a fitness experienced by a given phenotype in the emulsion for an arbitrary initial library. This allows to write the general update equation for the 
evolution of the library in course of selection. In Section~\ref{trajectory} we study some general properties of solutions to Cauchy problems for the 
derived update equation. We also obtain exact solutions for some special cases. Section~\ref{nonlinear} deals with generalizations to nonlinear 
phenotype-fitness dependencies, including polynomial functions and sums of exponentials. We also demonstrate how, at least in principle, to deal with an 
arbitrary continuous phenotype-fitness dependence using already established results. In Section~\ref{nonadditive} we outline the framework to capture 
more general situations like a non-Poissonian distribution of individuals in the compartments, nonadditive phenotype, and multiple traits. In 
Section~\ref{numerical} we provide the results of numerical simulations to test some predictions from previous sections. Finally, we briefly discuss 
the result and their relevance to biological situations. Technical introduction and detailed mathematical proofs are given in appendices.

\section*{List of notations}

The following notations are adopted throughout the article:

\begin{longtable}[h]{l p{0.2cm} l} 
\itshape Notation && \itshape meaning or comment\\
%\hline
\endfirsthead
\itshape Notation && \itshape meaning or comment\\
%\hline 
\endhead 
$\mathbb N$ && we assume $0 \in \mathbb N$\\
$\mathbb R_+$ && the nonnegative semiaxis: $\mathbb R_+ = [0,+\infty) \subset \mathbb R$\\
$C_c$ && space of continuous functions with compact support\\
$C_{c+}$ && space of nonnegative functions from $C_c$\\
$C'_c$ && space of generalized functions on $C_c$ (Radon measures)\\ 
$C'_{c+}$ && subset of nonegative generalized functions\\
$\mathbb P$ && subset of probability densities: $\mathbb P = \{\rho \in C_{c+}'\,|\,\langle \rho,1\rangle = 1\}$\\ 
$\mathbb P_p$ && finite point-mass densities: $\mathbb P_p = \{\rho \in \mathbb P\,|\,\rho = \sum\limits_{k=1}^n a_k 
	\delta_{x_k}$\}\\ 
$\mathcal I$ && some very large closed interval: $\mathcal I = [0,\mathcal L]$\\
$\mathbb P^\mathcal{I}$ && densities in $\mathcal I$: $\mathbb P^\mathcal{I} = \{\rho \in \mathbb P\,|\, \supp \rho \subset \mathcal I\}$\\ 
$\mathbb P^\mathcal{I}_p$ && finite point-mass densities in $\mathcal I$: $\mathbb P^\mathcal{I}_p = \mathbb P_p \cap \mathbb 
	P^\mathcal{I}$\vspace{0.2cm}\\ 
$\chi_A$ && indicator function of the set $A$: $\chi_A(x) = \begin{cases}1,&x \in A\\0,&x\notin A\end{cases}$\\ 
$C^k_n$ && binomial coefficient $\dfrac{n!}{k!(n-k)!}$\\ 
$\langle \rho, \phi \rangle$ & & the action of the generalized function $\rho$ on the test function $\phi$\\ 
$\langle \rho, \phi(x) \rangle$ & & implicitly $\langle \rho(x),\phi(x)\rangle$, where $x$ is the internal variable\\ 
$\langle \rho, \phi(x,y) \rangle$ \vspace{0.2cm} & & implicitly $\langle \rho(y),\phi(x,y)\rangle$, where $y$ is internal and $x$ is external\\ 
$\langle \rho_x, \phi(y) \rangle$ \vspace{0.2cm} & & \parbox{8.5cm}{implicitly $\langle \rho_x(y),\phi(y)\rangle$, where $y$ is the internal variable 
	and $x$ is a parameter of the distribution family $\{\rho_x\}$}\\
$g(x)$ && a shortcut for $(1 - e^{-x})/x$\\ 
$\delta_a$ && $\delta$-function concentrated at $a$: $\langle \delta_a, \phi \rangle = \phi(a)$\\ 
$\supp\phi$ \vspace{0.2cm} && support of the function $\phi$: the closure of $\{x \in \mathbb R\,|\,\phi(x) \neq 0\}$\\
$\supp\rho$ \vspace{0.2cm} && \parbox{8.5cm}{support of the generalized function $\rho$: $\supp \rho = \mathbb R \setminus O_\rho$, where $O_\rho$ is 
	the largest open subset $O \subset \mathbb R$ such that $\rho|_O = 0$}\\
$\bigotimes\limits_k \rho_k$ \vspace{0.2cm}&& tensor product $\rho_1 \otimes \rho_2 \otimes \ldots$\\ 
$\rho^{\otimes n}$ && $n$-th tensorial power: $\underbrace{\rho\otimes\rho\otimes\ldots\otimes\rho}_{n\text{ times}}$\\ 
$\Conv\limits_k \rho_k$ \vspace{0.2cm}&& convolution product $\rho_1 * \rho_2 * \ldots$\\ 
$\rho^{\ast n}$ \vspace {0.2cm} && $n$-th convolution power: $\underbrace{\rho\ast\rho\ast\ldots\ast\rho}_{n\text{ times}}$\\ 
$f_\star$ \vspace{0.2cm} && \parbox{8cm}{pushforward of a generalized function by the map $f$ of the domain: 
$\langle f_\star\rho, \phi\rangle = \langle \rho, \phi \circ f\rangle$}\\ 
$\mathrm{Corr}(\rho_1,\rho_2)$ \vspace{0.2cm} && cross-correlation of densities $\rho_1$ and $\rho_2$ \\
$\rho$ \vspace{0.2cm} && \parbox{8.5cm}{probability density of the phenotypes (in the model description and application)}\\
$\sigma$ \vspace{0.2cm}&& \parbox{8.5cm}{probability density of the fitness in a compartmentalized population}\\
$\sigma_x$ \vspace{0.2cm} && probability density of the fitness conditioned on phenotype $x$\\
$\bar x$\vspace{0.2cm} && \parbox{8.5cm}{mean phenotypic trait:
	mathematical expectation of the function $x\mapsto x$
	with respect to the phenotype distribution, $\langle \rho,x\rangle$ (in the model description and application)}\\ 
$\overline{x^n}$\vspace{0.2cm} && \parbox{8.5cm}{the $n$-th moment of the phenotypic trait:
	mathematical expectation of the function $x\mapsto x^n$
	with respect to the phenotype distribution, $\langle \rho,x^n\rangle$ (in the model description and application)}\\ 
$\bar \varw$ \vspace{0.2cm}&& \parbox{8.5cm}{mean fitness of an individual in a compartmentalized population: $\langle \sigma,x\rangle$
	(in the model description and application)}\\
$\bar \varw_x$ \vspace{0.2cm}&& \parbox{8.5cm}{mean fitness of an individual with pheontype $x$ in a compartmentalized population:
	$\langle\sigma_x,y\rangle$ (in the model description and application)}\\
$\ch x$ \vspace{0.2cm}&& hyperbolic cosine of $x$: $\ch x = (e^x + e^{-x})/2$\\
$\lambda$ \vspace{0.2cm}&& \parbox{8cm}{Poisson parameter: the mean number of individuals per compartment}\\
$\wedge$, $\Rightarrow$, $\neg$ && logical conjunction, implication, and negation, respectively
\end{longtable}

\section{Model}
\label{model}

We consider a population of haploid individuals that is subject to a group selection without any selection inside the groups. The population reproduces 
at discrete moments of time, different generations do not overlap, and the phenotype is strictly inherited. The groups model the compartments of the 
\emph{in vitro} evolution experiments. We will refer to groups as \emph{compartments} in the following. The term \emph{compartment} is more appropriate 
than the term \emph{group} because we can freely talk about empty compartments. The number of compartments $M$ and the population size $N$ are assumed 
to be very large, so we consider the infinite population limit and all the stochasticity at the population level is ignored. Both numbers are the same 
in all generations. At each generation, the compartments are entirely reformed and randomly repopulated from the pool of descendants produced by the 
previous generation. Following the infinite population limit, we assume that the number of individuals in each compartment is given by the Poisson 
statistics with the Poisson parameter $\lambda$, the average number of individuals per compartment (understood as a limit of $N/M$). Each individual 
produces as many descendants as any other member of the same compartment. The number of descendants of an individual can be called the \emph{(local) 
fitness of the individual}. It is defined by the total phenotype $x_\mathrm{tot}$ of the compartment in the following way. Firstly, we assume the total 
phenotype of the compartment to be additive. More specifically, if there are $n$ individuals in the compartment with phenotypes $x_1$, $x_2$, \ldots, 
$x_n$, then $x_\mathrm{tot} = x_1 + \ldots + x_n$. Secondly, each of the individuals produces $f(x_\mathrm{tot})/n$ descendants. Function $f$ (we call 
it the \emph{selection function}) bears all the specific information on the selection process by defining the overall replication/reproduction activity 
in compartments. The division by $n$ manifests the sharing of that replication activity by all the members of the current compartment. In the simplest 
case of linear selection we have $f(x_\mathrm{tot}) = x_\mathrm{tot}$ and the fitness of an individual is equal to the average phenotype in the 
compartment it belongs to. The phenotypic distribution of the descendants is obtained by averaging the result of the reproduction process over all 
compartments. We explicitly assume that the phenotypic distribution of the compartmentalized population at the next generation is equal to the 
phenotypic distribution of the descendants of the current generation. The corresponding lifecycle is schematically depicted on Figure~\ref{fig-scheme}.

\begin{figure}[t!]
\centering
\includegraphics{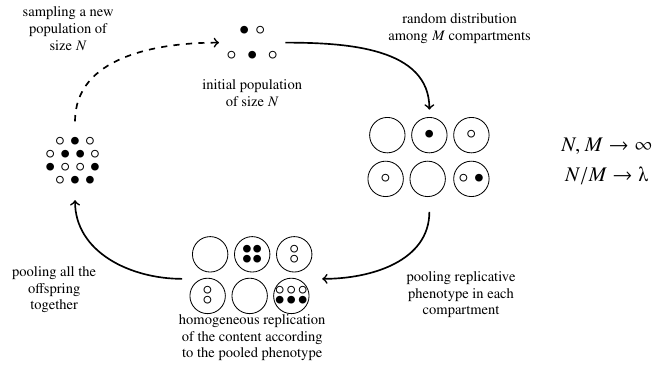}
\caption{
	The process of compartmentalized selection considered in the article. The initial population of size $N$ (here only two genotypes are shown, 
	$\circ$ and $\bullet$) is randomly encapsulated inside $M$ compartments. Some compartments contain multiple individuals with possibly different 
	genotypes. All individuals in a given compartment collectively contribute to the overall reproductive/replicative activity inside the compartment 
	by pooling together their reproductive/replicative phenotypes. The content of the compartment then reproduces such that the total number of new 
	individuals in the compartment depends on the collective phenotype \emph{via} the selection function $f$ (not shown, see the model description for 
	details). All individuals of the compartment equally contribute to this total progeny regardless of their genotype. Then the compartments are 
	broken and a new generation is formed by sampling from the pool of the offspring. Note that a genotype of a weaker phenotype ($\circ$) has an 
	opportunity to reproduce more efficiently if co-compartmentalized with a genotype of a stronger phenotype ($\bullet$). The article deals with the 
	deterministic infinite population model, when both $N \to \infty$ and $M \to \infty$ with the conservation of the average number of individuals per 
	compartment $\lambda = N/M$.
 }
\label{fig-scheme}
\end{figure}

One essential assumption that we will always imply is that the phenotypic distribution is compactly supported and the phenotypic values are 
non-negative. This assumption is technical. However, it is well justified by the application domain. Indeed, we are interested in the phenotype 
(related to the replication activity) as a random variable. The replication activity is non-negative and has some physical upper bound (the 
reproduction rate cannot be infinite nor can be the rate of an enzymatic reaction).

It should be noted that, unlike the classical haploid selection, a nonlinear compartmentalized selection in general cannot be reduced to the 
the linear selection by a simple nonlinear reparametrization of the phenotypic variable. The reason is that the phenotype additivity breaks with the 
reparametrization.

In the context of directed \emph{in vitro} evolution of enzymes in emulsions, the individuals would correspond to gene carrying constructs or other 
genetic vectors, the compartments would correspond to the droplets of the emulsion, while the phenotype would correspond, for example, to the enzymatic 
activity. In the context of structured population in a patched environment, one may think of the pool of descendants as of the dispersion phase that, 
on a new cycle, randomly recolonizes the same set of patches. In this case the compartments correspond to the ecological patches.

The philosophy of the article is the following. We first use a heuristic approach to derive the mathematical formulation of the problem. Then we study 
implications of this formulation in the full mathematical rigor.

\section{General theory of linear compartmentalized selection}
\label{selection-arbitrary}

In this section, we will consider only the linear selection function ($f(x) = x$). As it was written above, the linear selection corresponds to the 
case when the fitness of an individual is equal to the average phenotype in the compartment it belongs to. We will first develop a needed mathematical 
formalism. With this formalization, we will derive a general update formula that governs the dynamics of the phenotypic distribution. The update 
equation will be valid for any compactly supported phenotypic distribution, be it discrete, continuous, or more general than that. The equation will 
explicitly include $\lambda$ as a parameter and thus will allow to directly investigate the dependence of the selection dynamics on the degree of the 
co-compartmentalization. The approach and the notions used in this section will be generalized and reused in the following sections to treat more 
complicated cases of nonlinear selection functions and non-Poissonian compartmentalizations.

We will treat random variables using the formalism of generalized functions. We will use Sobolev's term \emph{generalized function} for a linear 
continuous functional on an appropriate space of test functions instead of Schwartz's term \emph{distribution} to avoid confusion with the somewhat 
ambiguous term \emph{probability distributions}. A random variable will be represented by its probability density, which in turn will be interpreted as 
a generalized function. This approach is equivalent to the standard treatment of random variables with the assumption of the weak convergence, where 
probability densities are interpreted as Radon measures. The emphasis is, however, put on the functional-theoretic aspect of the problem rather than on 
the measure-theoretic one. This point of view is more appropriate for the current analysis, as we are interested by computing average values of various 
functions rather than by finding probabilities of various events. Another advantage of the functional-theoretic approach is its intrinsic algebraicity, 
which strongly simplifies practical computations. See Appendix~\ref{generalized-functions} for a brief introduction to the subject.

Any physically realistic phenotypic distribution in a real population has only finite number of different values $x$ of the phenotype and, therefore, 
is represented by a point-mass probability density with finite number of $\delta$-functions. General probability densities enter the picture as 
approximation of these point-mass generalized functions when the number of points (phenotype classes) becomes unmanageably large. Therefore, it is 
natural to derive the update equation first on the subset of point-mass probability densities with finite number of points and then to extend it to all 
probability densities by continuity. We will denote $\mathbb P_p \subset \mathbb P$ the subspace of such point-mass densities, where $\mathbb P$ is the 
space of all probability densities. We want to derive an update equation in the form $\rho_{t+1} = A(\rho_t)$. To be able to extend the update operator 
$A$ to $\mathbb P$ or to some its subset by continuity from $\mathbb P_p$, we have to be sure that 1) $A$ is continuous on $\mathbb P_p$ with the 
subset topology and 2) that $\mathbb P_p$ is dense in $\mathbb P$. Unfortunately, such direct approach does not work because the operator $A$ 
constructed in the following appears not to be continuous on $\mathbb P_p$. However, an additional nonrestrictive assumption on the boundness of the 
possible phenotypic values removes this obstacle. The exact correct formulation of this assumption will be given in the end of this section.

Let us consider a random variable $\eta$ that describes the total phenotype of a compartment that contains $n$ individuals with phenotypes defined by 
$n$ random variables $\xi_1$, \ldots, $\xi_n$ with compactly supported densities $\rho_{\xi_k}$, so $\eta = \sum\limits_k \xi_k$. The joint density of 
the random ``vector'' defined by $\xi_k$ is given by $\rho_{(\xi_1\ldots\xi_n)} = \bigotimes\limits_k \rho_{\xi_k}$. Using the standard simplification 
due to the compactness of the supports of $\rho_k$, for any $\phi \in C_c$, by definition, we can express the action of $\rho_{(\xi_1\ldots\xi_n)}$ on 
$\phi\left(\sum\limits_k x_k\right)$ \emph{via} the convolution of the individual distributions
	\begin{equation}
	\langle \bigotimes_k\rho_{\xi_k}, \phi(x_1 + \ldots + x_n) \rangle = \langle \Conv_k \rho_{\xi_k}, \phi \rangle.
	\end{equation}

\noindent Therefore, $\rho_\eta = \Conv\limits_k\rho_{\xi_k}$. In case when $\forall k$ $\rho_{\xi_k} = \rho$, we have $\rho_\eta = \rho^{*n}$.

In the same way, we can compute the density of a random variable that corresponds to a per-individual fitness in a compartment with $n$ such 
individuals. This fitness is given by $\zeta = \eta/n$. Its density can be computed noting that it is a pushforward of $\rho_\eta$ with respect to the 
map $h_n\colon \mathbb R \to \mathbb R$, $x \mapsto x/n$, so
	\begin{equation}
	\rho_\zeta = (h_n)_\star \rho_\eta = (h_n)_\star \left(\Conv_k \rho_{\xi_k}\right).
	\end{equation}

\noindent Here one can understand $(h_n)_\star \rho\,(x)$ as $n\rho(nx)$ with a slight abuse of notations.

In the following we assume that all $\rho_{\xi_i}$ are the same and are equal to $\rho \in \mathbb P_p$, that is the phenotypic distribution in the 
population is characterized by the probability density $\rho$.

Let us find the density $\rho_{\zeta|k}$ of the per-individual fitness, given that the number of individuals in the compartment is $n$, and $k$ of 
$\xi_i$ assumed the value $x$ such that $\langle \rho, \chi_{\{x\}} \rangle = p > 0$, while the rest $n-k$ variables assumed any value different from 
$x$. This is equivalent to say, that first $k$ individuals are independently drawn from the distribution given by $\delta_x$, and the rest $n-k$ 
individuals are independently drawn from the distribution given by
	\begin{equation}
	\rho_{-x} \deq \frac{1}{1 - p}(\rho - p \delta_x).
	\label{rho-x}
	\end{equation}

\noindent We can immediately conclude that
	\begin{equation}
	\rho_{\eta|k} = \delta_x^{*k}*\rho_{-x}^{*n-k},
	\end{equation}

\noindent and then, as a consequence, the per-individual fitness has the probability density
	\begin{equation}
	\rho_{\zeta|k} = \left(n\delta_x^{*k}*\rho_{-x}^{*n-k}\right)(nx) = (h_n)_\star\left(\delta_x^{*k}*\rho_{-x}^{*n-k}\right)  \deq \sigma^x_{nk},
	\quad \text{where} \; h_n\colon x \mapsto x/n
	\label{sigma-n}
	\end{equation}

The next step is to find the density of fitness distribution of individuals with a given phenotype $x$ in the whole compartmentalized population, given 
that the phenotype of the initial library has density $\rho \in \mathbb P_p$. The density that we want to find describes the distribution of a local 
fitness (in a compartment) of an individual randomly chosen from all individuals with phenotype $x$. Let $P_n = e^{-\lambda} \lambda^n/n!$ be the 
probability to find a compartment with $n$ individuals (assuming the Poisson distribution with the mean number of individuals per compartments 
$\lambda$). As phenotype $x$ is present in macroscopic quantities in the population, the probability to find a compartment with $n$ individuals, $k$ of 
which are with phenotype $x$, and the rest have other values (an $nk$-class compartment), randomly drawing it from all compartments is given by
	\begin{equation}
	P_{nk} = \frac{e^{-\lambda}\lambda^n}{n!}C^k_n p^k (1-p)^{n-k} = P_nC^k_n p^k (1-p)^{n-k}, \quad \text{where} \; p = \langle \rho, \chi_{\{x\}}\rangle,
	\end{equation}

\noindent Each such compartment contains $k$ individuals with phenotype $x$, therefore the probability to find an individual that is encapsulated in an 
$nk$-class compartment randomly drawing it from the subpopulation of all individuals with phenotype $x$ is equal to
	\begin{equation}
	P^x_{nk} = \frac{k P_{nk}}{\sum\limits_{m,r} r P_{mr}}.
	\label{Pxnk}
	\end{equation}

The normalization constant in the denominator of (\ref{Pxnk}) is equal to
	\begin{multline}
	\sum_{n=1}^\infty\sum_{k=1}^n kP_{nk} = \sum_{n=1}^\infty \frac{e^{-\lambda}\lambda^n}{n!} \sum_{k=1}^n kC^k_n p^k (1-p)^{n-k} =
	\sum_{n=1}^\infty\frac{e^{-\lambda}\lambda^n}{n!}np = \lambda p.
	\end{multline}

In each $nk$-class compartment, the per-individual fitness density of each individual is given by $\sigma^x_{nk}$ from (\ref{sigma-n}), where 
$\rho_{-x}$ is defined by (\ref{rho-x}). Therefore, the fitness density $\sigma_x$ of an individual with phenotype $x$ randomly drawn from the whole 
compartmentalized population is given by
	\begin{multline}
	\sigma_x = \sum_{n=1}^\infty\sum_{k=1}^n P^x_{nk} \sigma^x_{nk} = \sum_{n=1}^\infty\sum_{k=1}^n\frac{kP_{nk}}{\lambda p} \sigma^x_{nk} =\\
	=\sum_{n=1}^\infty\frac{e^{-\lambda}\lambda^{n-1}}{n!}\sum_{k=1}^n k C^k_n p^{k-1}(1-p)^{n-k} 
	(h_n)_\star\left(\delta_x^{*k}*\rho_{-x}^{*n-k}\right)=\\
	= \sum_{n=0}^\infty\frac{e^{-\lambda}\lambda^n}{n!}\sum_{k=0}^n C^k_n p^k(1-p)^{n-k} 
	(h_{n+1})_\star\left(\delta_x^{*k+1}*\rho_{-x}^{*n-k}\right),
	\label{sigma-nk}
	\end{multline}

\noindent where again $h_n\colon x \mapsto x/n$.

To simplify this expression, we can use the linearity of $(h_i)_\star$, the facts that
	\begin{equation}
	\sum\limits_{k=0}^n C^k_n\, \rho_1^{*k}*\rho_2^{*n-k} = (\rho_1 + \rho_2)^{*n},
	\end{equation}

\noindent that $a(\rho_1 * \rho_2) = (a\rho_1)*\rho_2 = \rho_1*(a\rho_2)$, where $a$ is some number, and that $p\delta_x + (1-p)\rho_{-x} = \rho$. 
Finally, we obtain
	\begin{equation}
	\sigma_x = \sum_{n=0}^\infty\frac{e^{-\lambda}\lambda^n}{n!}(h_{n+1})_\star\left(\delta_x * \rho^{*n}\right).
	\label{sigma-x}
	\end{equation}

This formula remarkably depends neither on $p$, the probability to encounter the phenotype $x$ in the initial library, nor on $\rho_{-x}$. It is well 
defined even in the limit $p$ = 0. In fact, its structure resembles something expected for a continuous distribution, when the conditional probability 
to find another individual of phenotype $x$ in a compartment that already contains one is equal to 0.

In the same way we can compute the density of the fitness distribution for the whole population (without conditioning on the phenotype value):
	\begin{equation}
	\sigma = \sum_{n=0}^\infty \frac{e^{-\lambda}\lambda^n}{n!}(h_{n+1})_\star \rho^{*n+1}.
	\label{sigma}
	\end{equation}

The fact that $\sigma_x$ and $\sigma$ are indeed generalized functions, that is, that the generalized functional series used to define them (both 
(\ref{sigma-nk}) and (\ref{sigma-x}) for $\sigma_x$ and (\ref{sigma}) for $\sigma$) converge to some generalized functions, can be easily established 
using three facts: 1) the series have the form $\sum\limits_n a_n \rho_n$, where every $\rho_n$ is a compactly supported nonnegative generalized 
function that obeys $\langle \rho_n, 1 \rangle = 1$, 2) every $\phi \in C_c$ can be majorated by a constant in the sense that $\exists B_\phi \in 
\mathbb R_+$ $\forall x$ $|\phi(x)| \leqslant B_\phi$, and 3) the series $\sum\limits_n a_n$ converges absolutely. It is not difficult to see that 
these generalized functions are indeed probability densities with the support in the nonnegative semiaxis.

In the following, when the explicit dependence on time (understood as the population number) is required, we will denote it either as a subscript 
argument, like $\rho_t$, or with an argument in parentheses, like $\sigma_x(t)$. The latter means that we computed $\sigma_x$ in (\ref{sigma-x}) using 
the value of $\rho$ at generation $t$ (thus, using $\rho_t$). The same meaning will be implied for $\sigma(t)$ and for various expectations. We will 
omit the time argument when it is not important and no confusion is possible.

As the population is assumed to be infinite, the coefficient in front of each term in $\rho = \sum\limits_k p_k \delta_{x_k}$ deterministically changes 
according to $p_k(t+1) = p_k(t) \bar \varw_{x_k}(t) / \bar \varw (t)$ at any selection step, where $\bar \varw_{x_k}(t) = \langle \sigma_{x_k}(t) , x 
\rangle$ is the mean fitness of phenotypes $x_k$ at generation $t$ and $\bar \varw(t) = \langle \sigma(t), x\rangle$ is the mean fitness of the whole 
population at generation $t$.

The mean fitness in the linear selection case is simply given by
	\begin{equation}
	\bar \varw = \sum_{n=0}^\infty \frac{e^{\lambda}\lambda^n}{n!}\frac{(n+1)\langle \rho, x \rangle}{n+1} = 
	\langle \rho, x \rangle = \bar x,
	\label{mean-w}
	\end{equation}

\noindent so the population average fitness is exactly equal to the population average phenotype. This property is a specific feature of the phenotype 
additivity with the linear selection function $f(x) = x$ and with sharing of the reproduction/replicative activity in a compartment among the 
individuals in it. Indeed, under this condition, the local fitness of any individual in any compartment is exactly equal to the average phenotype in 
that compartment, and thus the result of (\ref{mean-w}) is intuitively expected.

The mean fitness of phenotypes $x_k$ is also easily computed
	\begin{equation}
	\bar \varw_{x_k} = \sum_{n=0}^\infty\frac{e^{-\lambda}\lambda^n}{n!}\frac{x_k + n\bar x}{n+1} = \bar x + g(\lambda)(x_k-\bar x).
	\label{xk-fitness}
	\end{equation}

\noindent Here the factor that depends on $\lambda$ is equal to
	\begin{equation}
	g(\lambda) \deq \frac{1 - e_{}^{-\lambda}}{\lambda}.
	\label{g-factor}
	\end{equation}

\noindent Function $g$ is monotonously decreasing with $g(0) = 1$ and $\lim\limits_{\lambda \to \infty} g(\lambda) = 0$. It is asymptotic to 
$1/\lambda$ at $\lambda \to +\infty$ (see Figure~\ref{g}).

\begin{figure}[t!]
\centering
\includegraphics{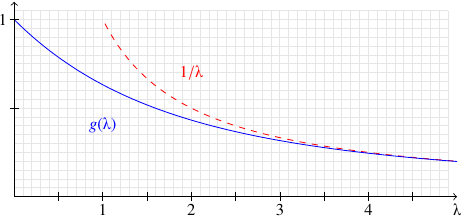}
\caption{The function $g(\lambda)$ and its asymptotic behaviour at large $\lambda$.}
\label{g}
\end{figure}

Note that the selection with the fitness given by (\ref{xk-fitness}) is frequency-dependent. The fitness of a phenotype does not depend only on 
properties of this phenotype but on properties of the whole population as total $\emph{via}$ the term with $\bar x$.

It follows that a given component $p_k \delta_{x_k}$ of $\rho$ after one round of selection changes to
	\begin{equation}
	p_k \delta_{x_k} \mapsto \left(1-g(\lambda) + g(\lambda)\frac{x_k}{\bar x}\right) p_k\delta_{x_k} = \left(1-g(\lambda) + g(\lambda)\frac{x}{\bar 
	x}\right) p_k\delta_{x_k}.
	\end{equation}

The update equation for the phenotypic distribution $\rho$ is then written as
	\begin{equation}
	\boxed{\rho_{t+1} = \left(1 - g(\lambda) + g(\lambda)\frac{x}{\bar x_t}\right) \rho_t,}
	\label{update-rho}
	\end{equation}

\noindent Where $x$ is the phenotypic value. Note that the influence of the Poisson compartmentalization comes only through the factor $g(\lambda)$ 
given by (\ref{g-factor}).

We can rewrite this update rule as $\rho_{t+1} = A(\rho_t)$. The update operator $A$ at the left-hand side is defined for any $\rho \in \mathbb P_p$ 
such that $\langle \rho , x\rangle \neq 0$. As was mentioned before, $A$ is not continuous even on $\mathbb P_p$ (see Appendix~\ref{app-continuity} 
Proposition~\ref{discontinuous}). However, as we demonstrate in Appendix~\ref{app-continuity}, for any closed interval $\mathcal I = [0,\mathcal L]$, 
this operator is continuous on the space of finite point-mass densities $\mathbb P^\mathcal{I}_p\setminus\{\delta_0\}$ with the support in $\mathcal I$ 
with the exclusion of the $\delta$-function concentrated at $x = 0$. We also demonstrate that $\mathbb P^\mathcal{I}_p$ is dense in $\mathbb 
P^\mathcal{I}$, the space of all probability densities with supports in $\mathcal I$. Therefore, $A$ can be extended by continuity to $\mathbb 
P^\mathcal{I}\setminus\{\delta_0\}$. Thus, any kind of general nonegatively and compactly supported probability density, that happens to well describe 
the library at hand, evolves according to (\ref{update-rho}). The operators that generate $\sigma_x$ in (\ref{sigma-x}) and $\sigma$ in (\ref{sigma}) 
are well defined and continuous on the whole $\mathbb P_p$. As $\mathbb P_p$ is dense in $\mathbb P$, they can be extended by continuity on all 
probability densities.

\section{Selection trajectory}
\label{trajectory}

In this section we will assume only the linear selection. The phenotypic density $\rho$ will be however considered to be any compactly supported 
probability density with $\supp \rho \subset \mathbb R_+$ and $\langle \rho , x \rangle \neq 0$.

Let the initial phenotypic distribution be given by $\rho_0$. We will consider its trajectory under the action of operator $A$ defined in the 
previous section. So, the update equation (\ref{update-rho}) is rewritten as $\rho_{t+1} = A(\rho_t)$ and we have $\rho_t = A^t(\rho_0)$.

Note that the action of $A$ on $\rho$ amounts to a multiplication of $\rho$ by the affine function with the slope $g(\lambda)/\langle \rho,x\rangle$ 
and with the intercept $1 - g(\lambda)$. Because of the dependence of parameters of this function on the current ${\bar x}_t$, (\ref{update-rho}) is 
not solvable in closed form for a generic $\rho_0$ and $\lambda > 0$ (see Figure~\ref{fig-selection}). Nevertheless, we can study some properties of a 
generic trajectory.

\begin{figure}[t!]
\centering
\includegraphics[scale=0.9]{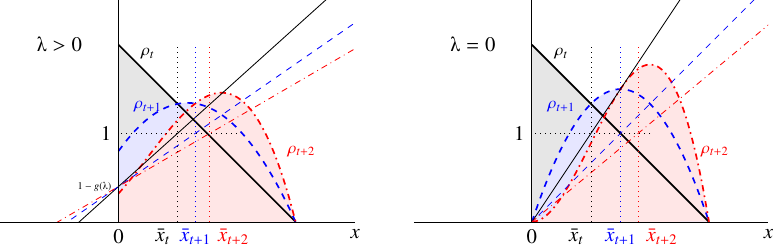}

\caption{The difference in the selection dynamics with $\lambda > 0$ and $\lambda = 0$. Phenotypic distribution densities $\rho_\tau$ at $\tau = t$, 
$\tau = t + 1$, and $\tau = t + 2$ are shown in thick solid line, thick dashed line, and thick dash-dotted line, respectively. Thin inclined straight 
lines are graphs of functions $1-g(\lambda) + g(\lambda)x/\bar x_\tau$, where $\tau = t$ (solid line), $t+1$ (dashed line), or $t+2$ (dash-dotted 
line). Note that the affine functions by which the operator $A$ multiplies the density at each step are different only by a rescaling for $\lambda = 
0$. In the case of $\lambda > 0$, these functions are different both by a rescaling and by a shift. This difference between the cases is responsible 
for the possibility to solve the case $\lambda = 0$ in closed form and for the lack of such solution for $\lambda > 0$. The initial density is taken to 
be $\rho_t = 2(1-x)\chi^{}_{[0,1]}$. On the left panel, $g(\lambda) = 0.6$, $\lambda \approx 1.15$.}

\label{fig-selection}
\end{figure}

\subsection{General properties of the trajectories}
\label{properties}

We will prove two generic properties.

\begin{property}
For any $\rho_0$ and any $t$ we have $\overline{x}_{t+1} - \overline{x}_t \geqslant 0$. In other words, the average phenotype does not 
decrease.
\label{property1}
\end{property}

\begin{proof}
To prove the first property it is enough to apply $\rho_{t+1}$ from (\ref{update-rho}) to $x$, which gives
	\begin{equation}
	\bar x_{t+1} = \bar x_t + g(\lambda)\frac{\overline{x^2}_t - \bar x_t^2}{\bar x_t}.
	\label{x-update}
	\end{equation}

\noindent The nonnegativeness of $\bar x_{t+1} - \bar x_t$ follows from the positiveness of $\bar x$ and the nonnegativeness of the variance 
$\overline{x^2}_t - \bar x_t^2$.  \end{proof}

Formula (\ref{x-update}) is reminiscent of Fisher's fundamental theorem of natural selection in Price's covariation form \cite{Price1970}. Indeed, if 
only the individual's phenotype mattered and was equal the individual's fitness, Price's formula for $\Delta \bar x_t \deq \bar x_{t+1} - \bar x_t$ 
would be written as
	\begin{equation}
	\Delta \bar x_t = \frac{\overline{x^2}_t - \bar x_t^2}{\bar x_t},
	\end{equation}

\noindent which is different from (\ref{x-update}) only by the factor $g(\lambda)$. Furthermore, any moment of the phenotypic distribution is updated 
according to
	\begin{equation}
	\Delta \overline{x^m}_t = g(\lambda)\frac{\overline{x^{m+1}}_t - \bar x_t \overline{x^m}_t}{\bar x_t},
	\label{xm}
	\end{equation}

\noindent which again transforms to Price's covariation formula for a deterministic phenotype-fitness relation in the limit $\lambda \to 0$.

\begin{property}
For any $\rho_0$ we have $\rho_t \to \delta_{x_0}$, as $t \to +\infty$, where $x_0 = \sup\supp\rho_0$. It 
means that the best mutant is always selected at infinite time.
\label{property2}
\end{property}

\noindent To prove the second property we first will prove an intuitive lemma.

\begin{lemma-main}
Let $\rho$ be a probability density with bounded support and the cardinality of $\supp \rho$ is greater than 1, then $x_i < \bar x < x_s$, where $x_i 
= \inf \supp \rho$ and $x_s = \sup \supp \rho$.
\end{lemma-main}

\begin{proof} 

Indeed, as $\rho$ is nonnegative, it is monotone in the following sense: $\forall \phi_1, \phi_2 \in C_c$ $\phi_1 \geqslant \phi_2$ $\Rightarrow$ 
$\langle \rho, \phi_1\rangle \geqslant \langle \rho, \phi_2\rangle$. To show that $\bar x \in [x_i, x_s]$, let us take test functions $\phi_i$, 
$\phi_x$, and $\phi_s$ such that $\phi_i|_{[x_i,x_s]} = x_i$, $\phi_x|_{[x_i,x_s]} = x$, and $\phi_s|_{[x_i,x_s]} = x_s$, and they are extended outside 
$[x_i,x_s]$ to respect $\phi_i \leqslant \phi_x \leqslant \phi_s$, which is always possible. Then, by monotonicity and by $\langle \rho, 1 \rangle = 
1$, we have $x_i \leqslant \langle \rho , x \rangle \leqslant x_s$.

Suppose that $\bar x = x_s$, so $\langle \rho,x \rangle = x_s$. Let $\phi \in C_c$ be a test function with the support in $(-\infty,x_s)$. Then there 
always exist a positive $a$ and a function $\phi_a \in C_{c+}$ such that $\phi_a|_{[x_i,x_s]} = a(x_s - x)$, $\phi_a(x) = 0$ for $x > x_s$, and 
$-\phi_a \leqslant \phi \leqslant \phi_a$. As $\langle \rho , \phi_a \rangle = 0$, it follows that $\langle \rho, \phi \rangle = 0$, and therefore, by 
definition $\supp \rho = \{x_s\}$, which contradicts the premise. In the same way we conclude that if $\bar x = x_i$, then $\supp \rho = \{x_i\}$. This 
proves the lemma. \end{proof}

\begin{proof}[of property \ref{property2}]
Solution of (\ref{update-rho}) at time $t$ can be written as
	\begin{equation}
	\rho_t = \rho_0\prod_{\tau = 0}^{t-1}\left(1 - g(\lambda) + g(\lambda)\frac{x}{\bar x_\tau}\right) \deq \rho_0 \Pi_t.
	\label{Pi}
	\end{equation}

\noindent For any $t$, $\rho_t$ is obtained from $\rho_0$ as its product with a positive (except possibly at 0, where it is zero for the special case 
$\lambda=0$) monotone continuous function $\Pi_t$. Therefore, $\supp \rho_t \subset \supp \rho_0$. More specifically, in most cases $\supp \rho_t = 
\supp \rho_0$. The only exception corresponds to $\lambda = 0$ and when $0 \in \supp \rho_0$ and it is not a limit point of the support. For this 
special case the following holds: $\forall t > 0$ $\supp \rho_t = \supp \rho_0 \setminus \{0\}$. Therefore, $\forall t$ $\sup \supp \rho_t = x_0$.

The monotonously increasing bounded sequence $\{\bar x_t\}$ has the limit $\bar x_\infty = \lim\limits_{t \to \infty}\bar x_t$ with $\bar x_\infty 
\leqslant x_0$. For any point $x_1$ such that $0 \leqslant x_1 < \bar x_\infty$, there exist $t_0$ such that for any $t \geqslant t_0$ we have 
$x_1/\bar x_t \leqslant x_1/\bar x_{t_0} < 1$. Let us denote $P_y^t(x) \deq (1 - g(\lambda) + g(\lambda)x/y)^t$. As $\Pi_t(x_1) \leqslant 
\Pi_{t_0}(x_1) P^{t-t_0}_{\bar x_{t_0}}(x_1)$ and $1 - g(\lambda) + g(\lambda)x_1/\bar x_{t_0} < 1$, we have $P^{t-t_0}_{\bar x_{t_0}}(x_1) \to 0$, and 
thus, $\Pi_t \to 0$ uniformly on $[0,x_1]$. It follows that $\rho_t|_{(-\infty,x_1)} \to 0$ with $t \to \infty$.

Suppose that $\bar x_\infty < x_0$. Then $\Pi_t \to \infty$ uniformly on any $(x_1,x_0] \subset [\bar x_\infty,x_0]$, and as $\supp \rho_t \cap 
(x_1,x_0] \neq \varnothing$, at large enough $t$ the relation $\langle \rho_t, 1 \rangle = \langle \rho_0, \Pi_t \rangle = 1$ is violated. Therefore, 
$\bar x_\infty = x_0$ and $\rho_t \to \delta_{x_0}$.
\end{proof}

\subsection{Solution in general and explicit solution for $\lambda = 0$}

For $\lambda > 0$, the solution of (\ref{update-rho}) is not representable with a simple closed formula. Nevertheless, it can be always computed by a 
simple recursion using (\ref{xm}) and (\ref{Pi}). We have $\rho_t = \Pi_t(x) \rho_0$, where $\Pi_0(x) = 1$, $\Pi_{t+1} = \left(1-g(\lambda) + 
g(\lambda)\dfrac{x}{\bar x_t}\right)\Pi_t$, and $\bar x_t$ is found by iterations $\overline{x^m}_t = (1 - g(\lambda))\overline{x^m}_{t-1} + 
g(\lambda)\dfrac{\overline{x^{m+1}}_{t-1}}{\bar x_{t-1}}$. As expected, $\Pi_t$ is a polynomial with coefficients that depend on all the moments of 
$\rho_0$ from $\bar x_0$ up to $\overline{x^{t+1}}_0$.

The case $\lambda = 0$, which effectively corresponds to the case when the number of compartments is much larger than the population size, can be 
solved in closed form. The idea of the following simple approach to solving (\ref{update-rho}) for this case was taken from \cite{Smerlak2017}. At 
$\lambda = 0$, $g(\lambda) = 1$ and the update equation (\ref{update-rho}) takes the simple form
	\begin{equation}
	\rho_{t+1} = \frac{x}{\bar x_t}\rho_t.
	\label{rho-equation-discr}
	\end{equation}

\noindent Note that, at every step, the density from the previous step is multiplied by $x$ and renormalized. Therefore, $\rho_t$ is proportional to 
$x^t\rho_0$ with some normalization constant, which is trivially reconstructed, and the solution is
	\begin{equation}
	\rho_t = \frac{\,x^t\,}{\overline{x^t}_0}\rho_0,\quad \bar x_t = \overline{x^{t+1}}_0/\overline{x^t}_0.
	\label{trajectory-discrete}
	\end{equation}

An interesting corollary of these formulas is the conclusion that the initial phenotypic probability density of the initial library can be 
reconstructed from the trajectory of $\bar x_t$ during the selection. Indeed, the following equality is a direct consequence of 
(\ref{trajectory-discrete}):
	\begin{equation}
	\overline{x^m}_0 = \prod_{i = 0}^{m-1} \bar x_i.
	\end{equation}

\noindent As the knowledge of all moments of $\rho_0$ allows to reconstruct the density itself, the phenotypic distribution in a library can be in 
principle obtained by tracking the mean population phenotype during selection process instead of direct measurement.

As an illustration, we will apply (\ref{trajectory-discrete}) to a library that contains only two classes of phenotypes, $x_1$ and $x_2$, and to a 
library that is described by a homogeneous distribution of phenotype on the interval $[x_1,x_2]$. In the former case, we have $\rho_0 = p\delta_{x_1} + 
(1-p)\delta_{x_2}$ and, thus, if we denote $\rho_t = p_t \delta_{x_1} + (1-p_t) \delta_{x_2}$, the solution is $p_t = \dfrac{p x_1^t}{p x_1^t + 
(1-p)x_2^t}$. In the latter case, $\rho_0 = \dfrac{1}{x_2 - x_1}\chi_{[x_1,x_2]}$ and, thus, $\rho_t = \dfrac{x^t(x_2^{t+1} - x_1^{t+1})}{(x_2 - x_1)^2 
(t+1)}\chi_{[x_1,x_2]}$.

\subsection{Exact solution for continuous time}

Unlike (\ref{update-rho}), its continuous time counterpart can be solved exactly for any value of $\lambda$. If we rewrite (\ref{update-rho}) as
	\begin{equation}
	\rho_{t+1} - \rho_t = g(\lambda)\left(\frac{x}{\bar x_t} - 1 \right)\rho_t,
	\end{equation}

\noindent and assume that at each selection step the changes are small (which is true, for example, for $\rho_0$ with small diameter of the support in 
comparison to $\bar x_0$, or for very large $\lambda$), then the selection dynamics can be approximated by
	\begin{equation}
	\frac{d\rho_t}{dt} = g(\lambda)\left(\frac{x}{\bar x_t} - 1\right)\rho_t.
	\label{rho-diff-full}
	\end{equation}

\noindent Here $\{\rho_t\}$ is understood as a $C^1$ one-parameter family of generalized functions and $d\rho_t/dt$ is understood as
	\begin{equation}
	\frac{d\rho_t}{dt} = \lim\limits_{\Delta t \to 0} \frac{\rho_{t + \Delta t} - \rho_t}{\Delta t}
	\end{equation}

\noindent in the topology of $C_c'$. We will prove the existence of solution of (\ref{rho-diff-full}) by construction.

Let us first assume that a solution of (\ref{rho-diff-full}) with given $\rho_0$ exists for any compactly supported phenotypic probability density 
$\rho_0$, $\supp \rho_0 \subset \mathbb R_+$. Then, for a given solution, we can assume $\bar x_t$ to be a given function of time that allows us to 
reparametrize $t$ in (\ref{rho-diff-full}) by the introduction a new independent variable $\tau$ such that
	\begin{equation}
	d\tau = g(\lambda)\frac{dt}{\bar x_t}.
	\label{t-tau}
	\end{equation}

\noindent This transforms (\ref{rho-diff-full}) to
	\begin{equation}
	\frac{d\rho_\tau}{d\tau} = (x - \bar x_\tau)\rho_\tau.
	\label{rho-diff-tau}
	\end{equation}

Let us introduce a new family of generalized functions $n_\tau$ that solves the following Cauchy problem
	\begin{equation}
	\frac{dn_\tau}{d\tau} = (x - 1)n_\tau,\quad n_0 = \rho_0.
	\label{n-diff}
	\end{equation}

One can look at $n$ as at the population density in the case when the normalization is not performed at the end of each 
selection cycle (the population is let to grow freely).

A solution of this problem corresponds to a solution of (\ref{rho-diff-tau}) by $\rho_\tau = \dfrac{n_\tau}{\langle n_\tau, 1 
\rangle}$. Indeed,
	\begin{equation}
	\frac{d\rho_\tau}{d\tau} = \frac{\dot n_\tau}{\langle n_\tau, 1 \rangle}
	- \frac{n_\tau \langle \dot n_\tau, 1 \rangle}{\langle n_\tau,1 \rangle^2} = (x - 1)\rho_\tau
	- \left(\frac{\langle xn_\tau,1\rangle}{\langle n_\tau,1\rangle} - 1\right)\rho_\tau = (x - \bar x_\tau) \rho_\tau,
	\end{equation}

\noindent where $\dot n_\tau = dn_\tau/d\tau$.

The solution of (\ref{n-diff}) exists, unique and is very easy to find. Namely, it is
	\begin{equation}
	n_\tau = e_{}^{x\tau - \tau} \rho_0.
	\end{equation}

\noindent The fact of its uniqueness can be established by passing to a Laplace transform of (\ref{n-diff}) with respect to $x$, which gives a linear 
PDE with constant coefficients, uniqueness of solution of which can be conventionally established by the method of characteristics and by the uniqueness 
and smooth dependence theorems for ODEs. A relevant condition here is the compactness of $\supp \rho_0$.

As $\langle n_\tau, 1 \rangle = \psi(\tau)e_{}^{-\tau}$, where $\psi(y) = \langle \rho_0, e_{}^{xy}\rangle$ is the so-called moment 
generating function of $\rho_0$ (it is its Laplace transform evaluated at $-\tau$), the corresponding solution of (\ref{rho-diff-tau}) is
	\begin{equation}
	\rho_\tau = \frac{e_{}^{x\tau}}{\psi(\tau)} \rho_0.
	\label{rho-tau-solution}
	\end{equation}

It follows that $\bar x_\tau = \psi'(\tau)/\psi(\tau)$, and thus, taking into account (\ref{t-tau}), the corresponding implicit solution of 
(\ref{rho-diff-full}) is given by
	\begin{equation}
	\rho_t = \frac{e_{}^{x\tau}}{\psi(\tau)} \rho_0, \quad \bar x_t = \frac{\psi'(\tau)}{\psi(\tau)}, 
	\quad t =\frac{\ln \psi(\tau)}{g(\lambda)}, \quad \psi(\tau) = \langle \rho_0, e_{}^{x\tau}\rangle.
	\end{equation}

This solution can be rewritten in an explicit form
	\begin{equation}
	\rho_t = e_{}^{x\psi^{-1}(e_{}^{g(\lambda) t}) - g(\lambda)t} \rho_0,\quad \bar x_t = e_{}^{-g(\lambda)t}\psi'(\psi^{-1}(e_{}^{g(\lambda)t})),
	\end{equation}

\noindent however, this explicit form is not practical, as for a generic case, the inverse of $\psi$ is impossible to compute explicitly. Even in a 
simple case of $\rho_0 = p\delta_{x_1} + (1-p)\delta_{x_2}$, for which $\psi(\tau) = p e_{}^{x_1 \tau} + (1-p)e_{}^{x_2 \tau}$, the inversion requires 
the solution of a transcendental functional equation.

The existence of a solution of (\ref{rho-diff-full}) follows from the explicit construction. It is also unique. Indeed, suppose that there are two 
solutions of (\ref{rho-diff-full}) $\rho^1_t$ and $\rho^2_t$ such that $\rho^1_0 = \rho^2_0 = \rho_0$. Then, with a solution dependent time rescaling, 
they both obey the same equation (\ref{rho-diff-tau}) with the same initial data but with possibly different $\tau$ (we denote them $\tau_1$ and 
$\tau_2$). The corresponding $n^i_{\tau_i}$ are uniquely constructed as $n^i_{\tau_i} = \psi(\tau_i)e_{}^{-\tau_i}\rho_0$, and it follows that the 
values of $t$ that correspond to $\tau_1$ and $\tau_2$ such that $\tau_1 = \tau_2$ are the same. Finally, we conclude that $\rho^1_t = \rho^2_t$.

The solution (\ref{rho-tau-solution}) of (\ref{rho-diff-tau}), along with the proof of its existence and uniqueness, was essentially obtained by 
different methods in \cite{Alfaro2014} for a class of regular probability densities (absolutely continuous distributions) and in \cite{Martin2016} for 
general densities. In particular, both works considered selection with mutations, while we left the question of mutations completely out of the scope. 
It is instructive to compare our explicit solution with the aforementioned results.

In \cite{Alfaro2014}, the exact solution was found for a more general ``replicator-mutator'' equation
	\begin{equation}
	\D_t u = (x - \bar x)u + a^2\D_{xx} u, \quad \bar x(t) \deq \intl{}{} xu(x,t)\,dx,
	\label{alfaro-eq}
	\end{equation}

\noindent where $u(x,t)$ is the phenotypic distribution and the Laplacian represents an approximation of mutation process. It is shown in this work 
that the exact solution to (\ref{alfaro-eq}) with the initial condition $u(x,0) = u_0(x)$ is given by
	\begin{equation}
	u(x,t) = \frac{\varv(x,t)}{1 + \intl{0}{t}\intl{}{} \xi \varv(\xi,\tau)\, d\xi d\tau}, \; \text{where }\varv\text{ solves}\; \D_t \varv = a^2\D_{xx} \varv 
	+ x\varv,\; \varv(x,0) = u_0(x).
	\end{equation}

\noindent The equation (\ref{rho-diff-tau}) with a regular $\rho_0$ corresponds to a degenerate case with $a = 0$. Then, $\varv$ is simply given by 
$\varv(x,t) = e^{xt}u_0(x)$ and $u(x,t) = e^{xt}u_0(x)/\intl{}{} e^{\xi t}u_0(\xi)\, d\xi$, which is identical to (\ref{rho-tau-solution}).

In \cite{Martin2016}, an approximate equation for the evolution of the cumulant generating function of the phenotypic distribution was studied for a 
wide class of mutation processes
	\begin{equation}
	\D_t C_t(\z) = \alpha(z)C_t'(\z) - C_t'(0) + \beta(\z),
	\label{martin-eq}
	\end{equation}

\noindent where $C_t(\z)$ is the expected cumulant distribution function under the deterministic approximation ($C(\z) = \ln \psi(\z)$), $\alpha$ and 
$\beta$ characterize the mutation process. This equation is equivalent to (\ref{rho-diff-tau}), if one assumes $\alpha(\z) = 1$, $\beta(\z) = 0$, and 
applies (\ref{rho-diff-tau}) on $e^{x\z}$ with the subsequent division by $\psi_t(\z)$ (note also that $C_t'(0) = \psi_t'(0) = \bar x_t$, as $\psi_t(0) = 
1$). It is demonstrated in \cite{Martin2016} that the solution of (\ref{martin-eq}) with the initial condition $C_0$ is given by
	\begin{equation}
	C_t(\z) = C_0\left(y\left(y^{-1}(\z) + t\right)\right) - C_0\left(y(t)\right) + \intL{0}{t}\left(\beta\left(y\left(y^{-1}(\z) + s\right)\right) - 
	\beta\left(y(s)\right)\right)\, ds,
	\end{equation}

\noindent where $y$ solves $y'(\z) = \alpha\left(y(\z)\right)$, $y(0) = 0$. In our case, $y(\z) = y^{-1}(\z) = \z$, and therefore $C_t(\z) = C_0(\z+t) - 
C_0(t)$, or $\psi_t(\z) = \psi_0(\z+t)/\psi_0(t)$, which is equivalent to (\ref{rho-tau-solution}).

\section{Nonlinear phenotype-fitness dependence}
\label{nonlinear}

In this section, we suppose that the local fitness of any individual inside a compartment with $n$ individuals is given by
	\begin{equation}
	h_n(x_1 + \ldots + x_n) = \frac{f(x_1 + \ldots +x_n)}{n},
	\label{h-nonadditive}
	\end{equation}

\noindent where the function $f$ associates the total number of new individuals generated in the compartment with the total phenotype in the 
compartment, and $x_i$ are phenotypes of the individuals in the compartment. In the following, we will call function $f$ a \emph{selection function}. 
The linear selection corresponds to the case $f = \mathrm{Id}_{\mathbb R}$.

If the dependence of the total in-compartment fitness on the total 
in-compartment phenotype is nonlinear, the update equation is still given by the formula
	\begin{equation}
	\rho_{t+1} = \dfrac{\langle\sigma_x(t),y\rangle}{\langle \sigma(t),y\rangle}\rho_t,
	\label{update-rho-nonlin}
	\end{equation}

\noindent where $\sigma_x$ and $\sigma$ are given by (\ref{sigma-x}) and (\ref{sigma}) with the only difference that $h_n$ is now taken to be $h_n(x) = 
f(x)/n $ as in~(\ref{h-nonadditive}). Its analysis, however, becomes much harder. The main problem is the inability to compute the evolution operator 
$A$ in closed form. Even in a case when each individual term in both $\langle \sigma_x , y \rangle$ and $\langle \sigma, y \rangle$ can be explicitly 
computed for any density $\rho$, it is difficult (and impossible in general) to find the explicit sum of these whole series. Another complication comes 
from the observation that although every term in these series is a compactly supported density, the support of the series themselves may happen to be 
unbound due to the expansion of the support of $(h_{n+1})_\star(\delta_x * \rho^{*n})$ and $(h_{n+1})_\star \rho^{*n+1}$ in case when $h_n$ grows 
sufficiently fast. This may lead to the divergence of the expectation for the fitness in the compartmentalized population. It is possible to show, 
however, that if $f$ is majorated by an exponential function $ae^{bx}$ (or more generally, by $a\ch bx$, which means that $f$ does not grow faster than 
some exponential from both sides), then not only are $\langle \sigma_x,y\rangle$ and $\langle \sigma,x\rangle$ well defined finite numbers, but also 
the maps $\rho \to \sigma_x$ and $\rho \to \sigma$ are continuous on $\mathbb P$ and the update operator $A$ is continuous on the set of probability 
densities with the support in some arbitrary large closed interval $\mathcal I$ (see Appendix~\ref{app-f-continuity}).

If the selection function $f$ is continuous, the update equation can be rewritten in different terms using the cross-correlation
	\begin{equation}
	\rho_{t+1} = \frac{1}{\mathcal N}\left(\sum_{n=0}^\infty \frac{e^{-\lambda}\lambda^n}{(n+1)!}\mathrm{Corr}(\rho_t^{*n},f)\right)\rho_t,
	\label{corr}
	\end{equation}

\noindent where $\mathcal N$ is the normalization constant and $\mathrm{Corr(\rho_1,\rho_2)}$ is the cross-correlation of densities $\rho_1$ and 
$\rho_2$. It is defined by its action on any function $\phi$ from $C_c$:
	\begin{equation}
	\langle\mathrm{Corr}(\rho_1, \rho_2), \phi\rangle = 
	\langle\rho_1(x_1),\langle\rho_2(x_2),\phi(x_2 - x_1)\rangle\rangle.
	\end{equation}

\noindent When one of $\rho_i$ is continuous, $\mathrm{Corr}(\rho_1,\rho_2)$ is continuous, too. Therefore, (\ref{corr}) is defined correctly for a 
continuous $f$. The only potential problem of this expression is when $n = 0$. We can define $\rho^{*0} = \delta_0$, which can be understood 
rigorously. Then $\mathrm{Corr}(\rho^{*0},f)(x) = \mathrm{Corr}(\delta_0,f)(x) = f(x)$.

If $\rho$ is also continuous, then one can use the classical formulas
	\begin{equation}
	\rho * \rho\,(x) = \intl{\mathbb{R}}{} \rho(y) \rho(x-y)\, dy,\quad \text{and}\quad
	\mathrm{Corr}(\rho,f)(x) = \intl{\mathbb{R}}{} \rho(y)f(x+y)\, dy.
	\end{equation}

We can also rewrite the update equation in a form that demonstrates its meaning more intuitively and transparently:
	\begin{equation}
	\rho_{t+1}(x) = \frac{1}{\mathcal N}\left(\sum_{n=0}^\infty \frac{e^{-\lambda}\lambda^n}{(n+1)!}\langle 
	f(x + x_1 + \ldots + x_n)\rangle^{}_{x_1,\ldots,x_n}\right)\rho_t(x),
	\end{equation}

\noindent where all $x_i$ are independently drawn from $\rho_t$ and we used an averaging notation common to physical and statistical literature.

In the following, we will consider two special cases that make it possible to compute $\rho_{t+1}$ based on $\rho_t$ in closed form, namely when $f$ is 
a polynomial and when $f$ is a linear combination of exponential functions.

\subsection{Polynomial $f$}
\label{polynomial}

Not only does this case allow to compute $\rho_{t+1}$ in closed form, but it also allows to derive the explicit expression for the update operator $A$, 
action of which is still a multiplication of $\rho_t$ by a polynomial function (coefficients of which nonlinearly depend on $\rho_t$).

Let $f(x) = a_0 + a_1 x + \ldots + a_m x^m$, $a_i \in \mathbb R$. Then every $h_n$ is a polynomial, too. Therefore, to compute $\langle \sigma_x, y 
\rangle$ and $\langle \sigma, y \rangle$, it is enough to find $(s^k_n)_x \deq \langle \delta_x*\rho^{*n}, y^k \rangle$ and $s^k_n \deq \langle 
\rho^{*n+1},y^k\rangle$ for any $k$, $0 \leqslant k \leqslant m$, and then to find the sums $\sum\limits_n\frac{P_n}{n+1}(s^k_n)_x$ and 
$\sum\limits_n\frac{P_n}{n+1}s^k_n$, where $P_n = e_{}^{-\lambda}\lambda^n/n!$. Indeed, for any density $\nu$ we have
	\begin{multline}
	\langle (h_{n+1})_\star \nu, x \rangle = 
	\frac{1}{n+1}\langle \nu, a_0 + a_1 x + \ldots + a_m x^m\rangle =\\=
	\frac{1}{n+1}(a_0 + a_1\langle \nu, x\rangle + \ldots + a_m\langle \nu,x^m\rangle).
	\end{multline}

The computation of the closed forms of $\langle \sigma_x, y\rangle$ and $\langle \sigma,y \rangle$, and therefore, of $A$, can in principle be
performed in an algorithmic way (see Appendix~\ref{app-polynomial}). As an example, we will treat the simplest nonlinear case $f\colon x \mapsto x^2$. 
There we have
	\begin{multline}
	\langle \delta_x* \rho^{*n}, y^2\rangle = \langle y^2(\delta_x*\rho^{*n}),1\rangle =
	\langle (y^2\delta_x)*\rho^{*n} + 2n(y\delta_x)*\rho^{*n-1}*(y\rho) + \\
	+n\delta_x*\rho^{*n-1}*(y^2\rho) + n(n-1)\delta_x*\rho^{*n-2}*(y\rho)^{*2},1\rangle =\\
	= x^2 + 2nx\bar x + n\overline{x^2} + n(n-1)\bar x^2.
	\end{multline}

\noindent Therefore
	\begin{equation}
	\langle \sigma_x,y\rangle = g_0x^2 + g_1\left(2x\bar x+ \overline{x^2}\right) + g_2 \bar x^2,
	\end{equation}

\noindent where $g_i = e_{}^{-\lambda}\lambda^i\dfrac{d}{d\lambda^i}\dfrac{e_{}^\lambda-1}{\lambda}$, so $g_0 = g(\lambda)$ (see 
Appendix~\ref{app-polynomial}). In the same way we find
	\begin{multline}
	\langle \rho^{*n+1}, y^2 \rangle = \langle y^2\rho^{*n+1},1\rangle = \langle
	(n+1)\rho^{*n}*(y^2\rho) + (n+1)n\rho^{*n-1}*(y\rho)^{*2}
	,1\rangle =\\=
	(n+1)\overline{x^2} + (n+1)n\bar x^2
	\end{multline}

\noindent and
	\begin{equation}
	\langle \sigma, y \rangle = \overline{x^2} + \lambda \bar x^2.
	\end{equation}

\noindent Finally, the update equation takes the form
	\begin{equation}
	\rho_{t+1} = \frac{g_0x^2 + g_1\left(2x\bar x_t+ \overline{x^2}_t\right) + g_2 \bar x_t^2}{\overline{x^2}_t + \lambda \bar x_t^2}\rho_t.
	\label{quadratic}
	\end{equation}

\subsection{$f$ as a linear combination of exponentials}

Another simple case is when the selection function $f$ is a linear combination of exponentials with constant coefficients, $f(x) = 
\sum\limits_{k=1}^m a_k e_{}^{b_k x}$. As in the case of a polynomial $f$, it is enough to compute $\langle \delta_x*\rho^{*n},e_{}^{by}\rangle$ and 
$\langle \rho^{*n+1},e_{}^{by}\rangle$ for an arbitrary $b$, which is, of course, a trivial problem. Indeed, $\langle 
\delta_x*\rho^{*n},e_{}^{by}\rangle = e_{}^{bx}\psi(b)^n$ and $\langle \rho^{*n+1},e_{}^{by}\rangle = \psi(b)^{n+1}$, where, as before, $\psi(s) \deq 
\langle \rho, e^{sy}\rangle$ is the moment generating function. The relevant sums are easily computable, too:
	\begin{eqnarray}
	\displaystyle
	\sum_{n=0}^\infty \frac{e_{}^{-\lambda}\lambda^n}{(n+1)!}e_{}^{bx}\psi(b)^n = \frac{e_{}^{\lambda \psi(b)} - 1}{\lambda \psi(b) 
	e_{}^{\lambda}}e_{}^{bx},\\
	\sum_{n=0}^\infty \frac{e_{}^{-\lambda}\lambda^n}{(n+1)!}\psi(b)^{n+1} = \frac{e_{}^{\lambda\psi(b)}-1}{\lambda e_{}^{\lambda}}.
	\end{eqnarray}
Therefore, the update equation takes the form
	\begin{equation}
	\rho_{t+1} = \frac{\displaystyle \sum\limits_{k=1}^m a_k \frac{e_{}^{\lambda\psi_t(b_k)}-1}{\psi_t(b_k)} e_{}^{b_k x}}{\displaystyle 
	\sum\limits_{k=1}^m a_k (e_{}^{\lambda\psi_t(b_k)}-1)}\,\rho_t.
	\label{exponential}
	\end{equation}

Interestingly, when $f(x) = e_{}^{bx}$, the selection dynamics does not depend on $\lambda$, as in this case $\rho_t = 
\dfrac{e_{}^{btx}}{\psi_0(bt)}\rho_0$.\\

Finally, the form of $f$ that has both polynomials and linear combinations of exponential with constant coefficients as its special cases is a linear 
combination of exponentials with polynomial coefficients
	\begin{equation}
	f(x) = \sum_{k=1}^m p_k(x) e_{}^{b_k x}, \quad \text{where}\quad p_k(x) = \sum_{j=0}^{n_k} a_{jk} x^j.
	\end{equation}

\noindent Using a combinations of the approaches considered in the current and the preceding sections, one can obtain the update equation for this 
function in closed form (see Appendix~\ref{app-polynomial}).

\subsection{Any continuous exponentially majorated $f$}

In principle, the action of $\rho$ can be extended on complex-valued test functions with values in $\mathbb C$ by a literal repetition of all the 
constructions outlined in Appendix~\ref{generalized-functions}. This gives a hope to use the result of the previous section to approximate an arbitrary 
$f$ by its truncated Fourier series on some large enough interval. As any such approximation is a sum of exponentials, the update equation for them can 
be written in closed form. In this case, all $b_k$ (in the notation of the previous section) are imaginary and $\psi(b_k t) = \phi(ib_k t)$, where 
$\phi$ is the characteristic function of $\rho$. However, one cannot, in general, select an interval once and then build approximations to any 
precision. The problem comes from the growing, in general, support of the growing convolution powers $\rho^{*n}$ in $\sigma_x$.

We demonstrate in Appendix~\ref{app-fourier} that for any continuous exponentially bound $f$ and for any compactly supported density $\rho$ with 
$\langle \sigma^f_x, y\rangle \neq 0$ there is a sequence of (possibly ever growing) intervals and trigonometric approximations $p_k$ of $f$ on these 
intervals based on the Fourier series that generates $\sigma^{p_k}_x$ and $\sigma^{p_k}$ such that $A^{p_k}(\rho) \to A^f(\rho)$, $\sigma^{p_k}_x \to 
\sigma^f_x$, and $\sigma^{p_k} \to \sigma^f$, where $\sigma^F_x$, $\sigma^F$, and $A^F$ denote the corresponding $\sigma_x$, $\sigma$, and $A$ 
generated by a selection function $F$ and the density $\rho$.

\section{Generalization to a non-Poissonian distribution in compartments, to nonadditive phenotype, and to multi-trait phenotype cases}
\label{nonadditive}

In this section, we do not aim for the proof of existence, continuity, etc of the relevant operators and operations or for the study of the generality 
of the results. We will just point out a possible generalization of the framework developed so far.

\subsection{Non-Poissonian distribution}
\label{non-poissonian}

It is possible to generalize $\sigma_x$ and $\sigma$, and thus the update operator, to an arbitrary repartition of individuals in the compartments. A 
non-Poissonian distributions can arise, for example, if bacteria are used as an intermediate vehicle for a genome and its product proteins before the 
compartmentalization. If the bacteria have a tendency to stick to each other it would in this case disturb the Poisson distribution of the bacteria in 
the compartments. In any case, this deviation is expressed in the fact that the probability to find a compartment with $n$ individuals $P_n$ is 
different from $e_{}^{-\lambda}\lambda^n/n!$. Following the same path as in Section \ref{selection-arbitrary}, it can be demonstrated that for 
arbitrary $P_n$ we have
	\begin{equation}
	\sigma_x = \frac{\displaystyle\sum_{n=1}^\infty n P_n (h_n)_\star (\delta_x * \rho^{*n-1})}{\displaystyle\sum_{n=1}^\infty n P_n},\quad\quad
	\sigma = \frac{\displaystyle\sum_{n=1}^\infty n P_n (h_n)_\star \rho^{*n}}{\displaystyle\sum_{n=1}^\infty n P_n}.
	\end{equation}

Of course, it is enough to know only $\sigma_x$ and $\langle \sigma,x \rangle$ can always be computed as $\langle \langle \sigma_x, y \rangle \rho, 1 
\rangle$. However, as it is seen from practical examples, a separate computation of $\sigma$ can be simpler. In addition, $\sigma$ has its own physical 
significance.

The update equation is, as before, $\rho_{t+1} = \dfrac{\langle\sigma_x(t),y\rangle}{\langle\sigma(t),y\rangle}\rho_t$. The cross-correlation form of 
it is represented by
	\begin{equation}
	\rho_{t+1} = \frac{1}{\mathcal N}\left(\sum_{n=1}^\infty P_n\mathrm{Corr}(\rho_t^{*n-1},f)\right)\rho_t,
	\end{equation}

\noindent where $\mathcal N$ is the normalization constant.

As before, we provide a form of this equation written in conventional notations of physical and applied statistical literature:
	\begin{equation}
	\rho_{t+1}(x) = \frac{1}{\mathcal N}\left(\sum_{n=1}^\infty P_n\langle f(x + x_1 + \ldots + 
	x_{n-1})\rangle^{}_{x_1,\ldots,x_{n-1}}\right)\rho_t(x),
	\end{equation}

\noindent where all $x_i$ are independently drawn from $\rho_t$.

\subsection{Nonadditive phenotype}

The assumption of the additivity of the phenotype, although reasonable in some cases, is far from being universal. For example, if the fitness value of 
a compartment is defined by some enzymatic kinetics, say, Michaelis-Menten reaction, the effective kinetic parameters may not be additive. If the 
fitness is defined by the total enzymatic reaction rate (encoded by the genomes of individuals) under assumption of the excess of the substrate, then we 
have the additivity. However, if enzymes belonging to different phenotypes have different affinity to the substrate and the substrate is not in excess, 
the additivity is lost.

The most general way to describe a nonadditive phenotype is to declare what happens to the local fitness of individuals in a compartment when different 
number of individuals with different phenotypes are enclosed together. This implies a definition of a family of selection functions $\{f_n\}$, $f_n 
\colon \mathbb R^n \to \mathbb R$, so if a compartment contains $n$ individuals with individual phenotypes $x_1$, \ldots, $x_n$, then the 
per-individual fitness in the compartment is given by $f_n(x_1,\ldots, x_n)$. When the local fitness of individuals is defined by the total phenotype 
of their compartment, functions $f_n$ have the form of compositions $f_n = h_n \circ f \circ c_n$ of the sharing function $h_n \colon x \mapsto x/n$, 
of the selection function $f\colon \mathbb R \to \mathbb R$ that maps the total phenotype in the compartment to the total number of offspring, and of 
the ``combination'' functions $c_n \colon \mathbb R^n \to \mathbb R$ that define the total phenotype from the phenotypes of individuals. Even more 
general case, when $f_n$ do not have the structure of the composition, corresponds to situations, when the very notion of the total phenotype is 
inapplicable.

All functions $f_n$ are naturally symmetric in the sense that if we denote by $\varpi$ a permutation of the set $\{1,\ldots,n\}$, then
	\begin{equation}
	\forall \varpi\quad f_n(x_{\varpi(1)},\ldots,x_{\varpi(n)}) = f_n(x_1,\ldots,x_n).
	\end{equation}

\noindent This symmetry comes from the fact that all individuals in a compartment are equivalent and it does not matter which one we consider to be the 
first one, which one to be the second one, and so on. If the compartments had some internal structure, or the total phenotype depended on the order of 
entry of individual (think of the order of infection of the same bacterium by multiple phages), then the symmetry would be lost.

Let us first consider $\rho \in \mathbb P_p$. Then the probability density of the local fitness of individuals with phenotype $x$ in an $nk$-class 
compartment, analogously to (\ref{sigma-n}), is given by
	\begin{equation}
	\sigma^x_{nk} = (f_n)_\star \left( \delta_x^{\otimes k} \otimes \rho_{-x}^{\otimes n-k} \right),
	\end{equation}

\noindent with the same notations as in Section~\ref{selection-arbitrary}. Here $(f_n)_\star \colon C_c'(\mathbb R^n) \to C_c'(\mathbb R)$ is the 
pushforward generated by $f_n$.

Note that for any $\rho_1$, \ldots, $\rho_n$ and any permutation $\varpi$ we have
	\begin{equation}
	(f_n)_\star(\rho_{\varpi(1)} \otimes \ldots \otimes \rho_{\varpi(n)}) = (f_n)_\star(\rho_1 \otimes \ldots \otimes \rho_n).
	\end{equation}

Indeed, for any $\phi \in C_c(\mathbb R)$ we have 
	\begin{multline}
	\langle (f_n)_\star(\rho_1 \otimes \ldots \otimes \rho_n) , \phi \rangle =
	\langle \rho_1(x_1) \otimes \ldots \otimes \rho_n(x_n) , \phi\big(f_n(x_1,\ldots,x_n)\big) \rangle = \\=
	\langle \rho_1(x_1), \langle \ldots, \langle \rho_n(x_n) , 
	\phi\big(f_n(x_{\varpi(1)},\ldots,x_{\varpi(n)})\big) \rangle \ldots \rangle \rangle = \\=
	\langle \rho_{\varpi(1)}(x_{\varpi(1)}), \langle  \ldots, \langle \rho_{\varpi(n)}(x_{\varpi(n)}) , 
	\phi\big(f_n(x_{\varpi(1)},\ldots,x_{\varpi(n)})\big) \rangle\ldots\rangle\rangle = \\=
	\langle (f_n)_\star(\rho_{\varpi(1)} \otimes \ldots \otimes \rho_{\varpi(n)}) , \phi \rangle
	\end{multline}

Because of this and because $(f_n)_\star$ is a linear operator, the following binomial identity is valid
	\begin{equation}
	\sum_{k=0}^n C^k_n p^k (1 - p)^{n-k} (f_n)_\star \left( \delta_x^{\otimes k} \otimes \rho_{-x}^{\otimes n-k} \right) = 
	(f_n)_\star \rho^{\otimes n}.
	\end{equation}

\noindent Here we identify $\rho^{\otimes 0}$ with $1 \in \mathbb R$, $\rho \otimes 1$ and $1 \otimes \rho$ with $\rho$.

Following the same reasoning as in Section~\ref{selection-arbitrary} and in Section~\ref{non-poissonian}, we can conclude that the fitness 
densities $\sigma_x$ and $\sigma$ are given by
	\begin{equation}
	\sigma_x = \frac{\displaystyle\sum_{n=1}^\infty n P_n (f_n)_\star (\delta_x \otimes \rho^{\otimes n-1})}{\displaystyle\sum_{n=1}^\infty n 
	P_n},
	\quad\quad
	\sigma = \frac{\displaystyle\sum_{n=1}^\infty n P_n (f_n)_\star \rho^{\otimes n}}{\displaystyle\sum_{n=1}^\infty n P_n},
	\end{equation}

\noindent or, in the case of Poissonian $P_n = e^{-\lambda} \lambda^n/n!$,
	\begin{equation}
	\sigma_x = \sum_{n=0}^\infty \frac{e^{-\lambda}\lambda^n}{n!} (f_{n+1})_\star (\delta_x \otimes \rho^{\otimes n}),
	\quad\quad
	\sigma = \sum_{n=0}^\infty \frac{e^{-\lambda}\lambda^n}{n!} (f_{n+1})_\star \rho^{\otimes n+1}.
	\end{equation}

These expressions can be extended by continuity to all $\rho \in \mathbb P$. The proof of continuity is analogous to the proof of 
Theorem~\ref{sigma-continuous} in Appendix~\ref{app-continuity}. The only essential additional fact to be used is the continuity of the tensor product. 
The continuity of the update operator is a more delicate issue and we will not study it here.

As before, we provide a conventional intuitively transparent form of the corresponding update equation for a general $P_n$ (all $x_i$ are independently 
drawn from $\rho_t$):	
	\begin{equation}
	\rho_{t+1}(x) = \frac{1}{\mathcal N}\left(\sum_{n=1}^\infty n P_n\langle f_n(x,x_1,\ldots,x_{n-1})\rangle^{}_{x_1,\ldots,x_{n-1}}\right)\rho_t(x).
	\end{equation}

\subsection{Multiple traits}
\label{multitrait}

In the same way we can consider not only a single trait phenotype $x$, but also multiple traits $x_1$, \ldots, $x_m$. We can organize them in a tuple 
$\xi = (x_1,\ldots,x_m) \in \mathbb R^m \deq \mathbb X$. The distribution of the traits at the beginning of each selection cycle is now given by a 
generalized function from $\mathbb P(\mathbb X) \subset C_c'(\mathbb X)$.

In the classical case, when no co-compartmentalization occurs, the selection process is ignorant about how exactly the fitness is defined by the 
underlying traits. Therefore, we could abstract from these traits once by using a pushforward of the trait distribution with the selection function. If 
the fitness is given by the selection function on the trait space $f\colon \mathbb X \to \mathbb R$ and the initial trait distribution is given by 
$\rho \in \mathbb P(\mathbb X)$, then the initial fitness distribution is given by $\tilde\rho = f_\star \rho \in \mathbb P(\mathbb R)$. As the 
resulting fitness distribution after $t$ cycles of selection is given by the initial fitness distribution multiplied by some continuous function of the 
fitness values, namely by $\tilde \rho_t = \dfrac{x^t}{\langle \tilde\rho, x^t \rangle}\tilde\rho$ (see (\ref{trajectory-discrete})), the final trait 
distribution could be reconstructed as the initial trait distribution multiplied by the pullback of this function by the selection function. Indeed, 
the pullback of any $\phi \in C_c(\mathbb R)$ to $C_c(\mathbb X)$ by $f$ is defined by $f^\star \phi \deq \phi \circ f \in C_c(\mathbb X)$. Using the 
identity $\phi (f_\star \rho) = f_\star\big((f^\star \phi)\rho\big)$, which is easy to verify, one recovers $\rho_t = \dfrac{f(\xi)^t}{\langle \rho, 
f(\xi)^t\rangle}\rho$. This expression goes well with the intuition (see the derivation of (\ref{trajectory-discrete})).

Unfortunately, the co-compartmentalization complicates this picture. As previously, we define the rules of how the phenotypic parameters (the traits) 
are related to the local fitness of individuals by a family of functions $\{f_n\}$ on Cartesian powers of the trait space $f_n \colon \mathbb X^n \to 
\mathbb R$. Functions $f_n$ assign the per-individual fitness $f_n(\xi_1,\ldots,\xi_n)$ to a compartment with $n$ individuals with phenotypes $\xi_1$, 
\ldots, $\xi_n$. In a case of a non-additive but still well defined total phenotypic traits in a compartment, each function $f_n$ is in fact a 
composition $f_n = h_n \circ f \circ c_n$ of the sharing functions $h_n \colon x \mapsto x/n$, the selection function $f \colon \mathbb X \to \mathbb 
R$ that operates on the combined phenotype of the compartment, and of the functions $c_n \colon \mathbb X^n \to \mathbb X$ that define how individual 
phenotypic traits are combined when multiple individuals are mixed in one compartment. In a more general case $f_n$ do not have this structure.
 
Analogously to the single trait case, we can start from $\mathbb P_p(\mathbb X)$ and obtain the expressions
	\begin{equation}
	\sigma_\xi = \frac{\displaystyle\sum_{n=1}^\infty n P_n (f_n)_\star (\delta_\xi \otimes \rho^{\otimes n-1})}{\displaystyle\sum_{n=1}^\infty n 
	P_n},
	\quad\quad
	\sigma = \frac{\displaystyle\sum_{n=1}^\infty n P_n (f_n)_\star \rho^{\otimes n}}{\displaystyle\sum_{n=1}^\infty n P_n},
	\end{equation}

\noindent or, in the case of Poissonian $P_n = e^{-\lambda}\lambda^n/n!$,
	\begin{equation}
	\sigma_\xi = \sum_{n=0}^\infty \frac{e^{-\lambda}\lambda^n}{n!} (f_{n+1})_\star (\delta_\xi \otimes \rho^{\otimes n}),
	\quad\quad
	\sigma = \sum_{n=0}^\infty \frac{e^{-\lambda}\lambda^n}{n!} (f_{n+1})_\star \rho^{\otimes n+1}.
	\end{equation}

These formulas are again extended by continuity to all $\rho \in \mathbb P(\mathbb X)$. Note that $\sigma_\xi$ and $\sigma$ are probability densities 
on the space of the per-individual fitness $\sigma_\xi, \sigma \in \mathbb P(\mathbb R)$. The update operator $A \colon \mathbb P(\mathbb X) \to 
\mathbb P(\mathbb X)$ is given by $A(\rho) = \dfrac{\langle \sigma_\xi, x \rangle}{\langle \sigma, x \rangle}\rho$. We will discuss neither its domain 
nor the condition of its continuity.

\section{Numerical simulations}
\label{numerical}

\subsection{Linear selection with large number of compartments and large population}

\begin{figure}[t!]
\centering
\includegraphics[scale=0.3]{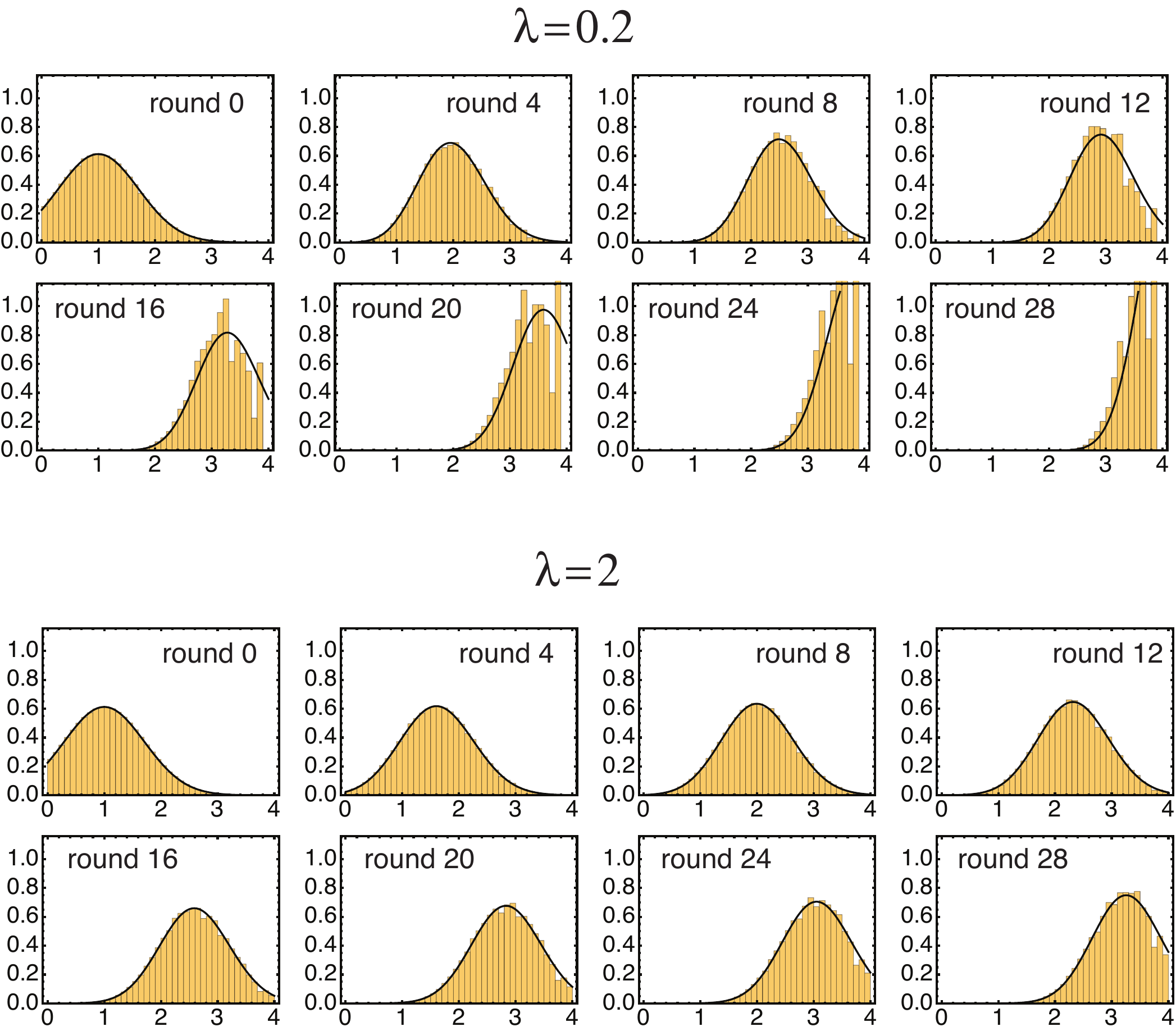}

\caption{Numerical simulations and their comparison to the theory. Individual trajectories are shown for two different values of $\lambda$. The 
phenotypic distribution in the population at each round of selection is shown by the histogram. The theoretical prediction obtained by iterating 
equation~(\ref{update-rho}) is shown as the solid line. The initial phenotypic distribution in each case is $\rho(x) = 
\frac{1}{\mathcal N}e^{-(x-1)^2}\chi_{[0,4]}(x)$, where $\mathcal N$ is the normalization constant. The horizontal axis represents the phenotypic value.}

\label{simulations} 
\end{figure}

To test the prediction given by the update equation (\ref{update-rho}) for the additive linear case, we carried out numerical simulations of the 
compartmentalized selection. The simulations were performed with Wolfram Mathematica. The corresponding notebook is provided in Online Resource. The 
number of compartments was fixed at $10^6$. An initial set of $2\cdot10^5$ or $2\cdot10^6$ phenotypes (depending on $\lambda$) was then drawn from a 
Gaussian distribution (centered at 1, with variance $1/2$) on the interval $[0,4]$.

One generation is implemented using the following loop, which is then repeated $n$ times.
\begin{itemize}

\item Each value from the set is randomly assigned to a compartment.

\item Each compartment is given the local fitness as the mean of the encapsulated phenotypes (with rounding, when needed).

\item An updated weight for each phenotype value is obtained by summing the local fitness of each compartment in which it was present (taking into 
account the multiplicity of that phenotype in each compartment, i.e. having in the end the number of offspring it was able to create overall).

\item A new set of $2\cdot10^5$ or $2\cdot10^6$ phenotype values (depending on $\lambda$) is drawn randomly from the list of the present phenotypes, 
using the updated weight list.

\end{itemize}

The result and its comparison to the theoretical prediction are shown on Figure~\ref{simulations}. One can see that the agreement is very good.

\subsection{The limit of small number of compartments and small population}

Although the predictions of the deterministic theory agree well with the individual-based simulation for reasonably large populations (populations of 
$10^5$--$10^6$ or more are typical for viral infections and for directed evolution experiments in emulsions), we also investigated their applicability 
to small populations and to small number of compartments. In a standard selection scenario (like the Fisher-Wright process), the only relevant effect 
observed in small populations is the genetic drift---a stochastic deviation from the deterministic selection due to random sampling effect in the 
construction of a new generation of a finite number of individuals. The compartmentalization adds another source of stochasticity, and thus it is 
interesting to compare these two cases. This additional randomness comes from a probabilistic connection between the phenotype of an individual and its 
fitness. Different realizations of the finite population packing into compartments result in different fitness assignments to individuals. Additional 
complication comes from the fact that this fitness assignment is not independent for different individuals. This implies a possible difference between 
the stochastic selection in a finite population with and without compartmentalization. As the formal treatment of the compartmentalized case in a 
finite population is very difficult and deserves a separate investigation, we compared the two cases by numerical simulations reducing both the 
population size $N$ (relevant for both cases) and the number of compartments $M$ (which may give a nonnegligible effect).

\begin{figure}
\centering
\includegraphics[scale=0.45]{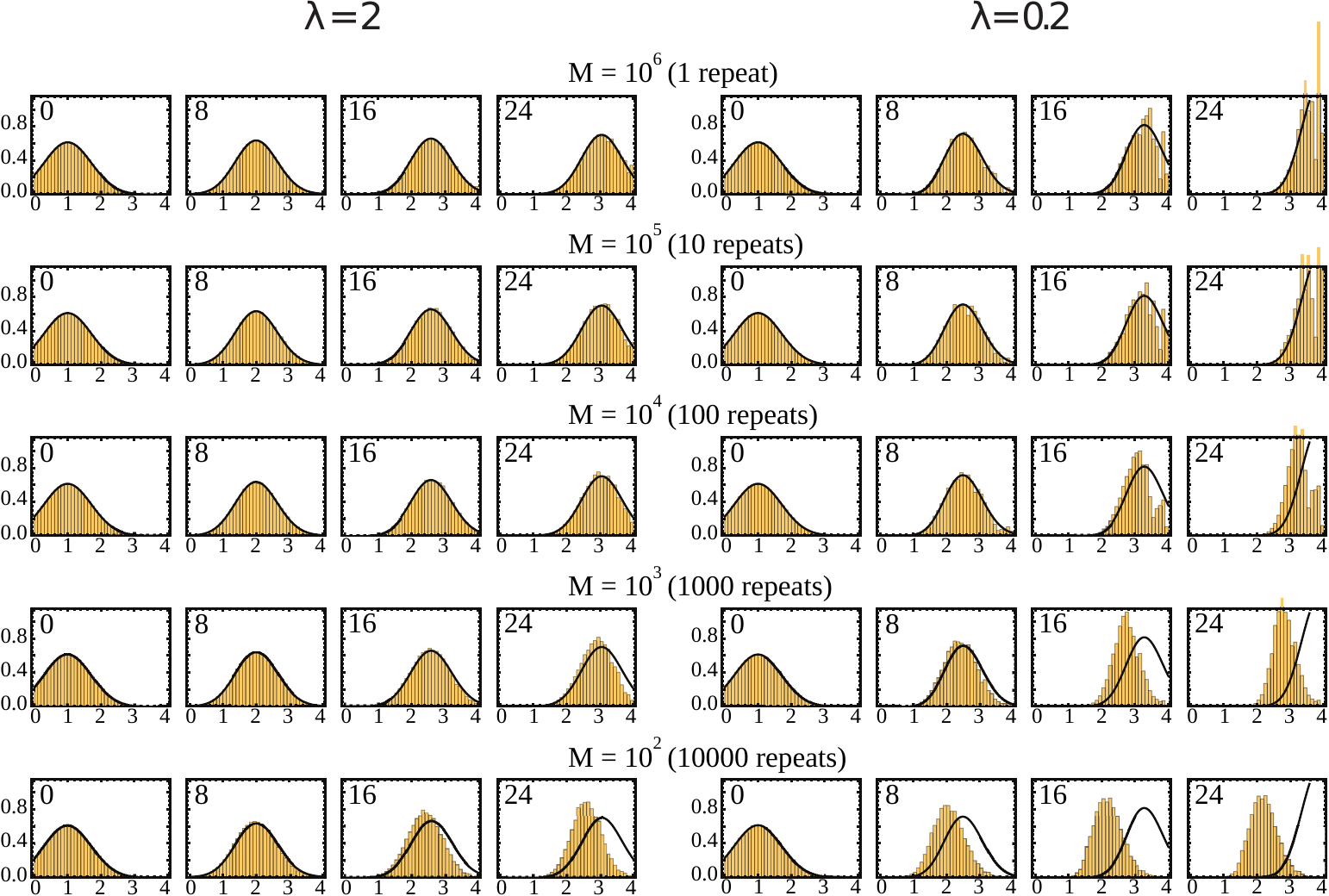}
\caption{
	The same numerical experiments (the same initial distribution, the same selection algorithm) as on Figure~\ref{simulations} performed with 
	different number of compartments $M$ and for two values of $\lambda$ for each value of $M$. The ensemble average phenotypic distributions are shown 
	for each indicated number of repeats. The black line shows the theoretical prediction by equation~(\ref{update-rho}). Numbers in boxes indicate 
	rounds of selection. The horizontal axis represents the phenotipic value.
}
\label{simulationsMN-all}
\end{figure}

Individual realizations of simulations become more and more random as $N$ decreases, both because of strong genetic drift of phenotypes with low 
frequencies and because of randomness in the initial drawing from the continuous distribution. This is especially important for the high-value tail of 
the initial phenotypic distribution. Indeed, the irregularities due to the randomness are already seen on Figure~\ref{simulations}. And as the high 
phenotype value front of the distribution is crucially important for the result of long-term selection, effectively defining the final state, it makes 
more sense to look at the average of an ensemble of trajectories instead of individual trajectories, when the stochasticity becomes too strong. 
Figure~\ref{simulationsMN-all} shows the ensemble average of individual-based simulations (phenotypic distributions are averaged for each time frame) 
for various $N$ and $\lambda$ (and thus for various $M$, as $\lambda = N/M$) for the same simulation algorithm and for the first generation drawn from 
the same phenotypic distribution, as in the previous section. To keep the statistics comparable, for each $M$ we ran simulations $m$ times to keep 
$mM = 10^6$ constant for values of $M$ $10^6$, $10^5$, \ldots, $10^2$. As before, we used $\lambda = 2$ and $\lambda = 0.2$ for every given $M$, thus 
changing $N$ from $2\cdot10^6$ to $20$.

The limit distribution of the averaged trajectory for small populations observed on these results is the distributions of the final fixed phenotypes of 
individual trajectories. This spreading can have three sources: 1) the initial sampling from the tail of $\rho_0(x) = \frac{1}{\mathcal 
N}e^{-(x-1)^2}\chi_{[0,4]}(x)$, which may result in the maximal value phenotype in the sampling smaller than 4, 2) the standard genetic drift, which 
may result in the fixation of a phenotype different from the maximal one in the initial sampling, and 3) the stochastic effects due to 
compartmentalization.

\begin{figure}
\centering
\includegraphics{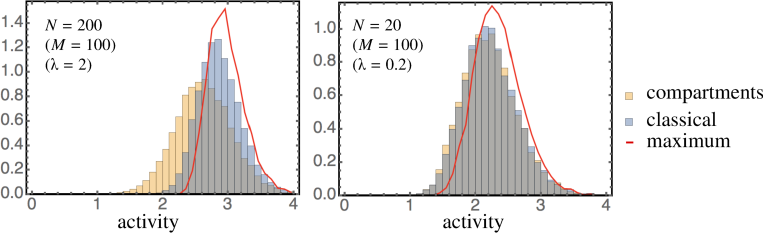}
\caption{
	Distribution of fixed phenotypes in individual-based simulations of compartmentalized and classical linear selection for two small values of the 
	population size. The corresponding $M$ and $\lambda$ for the compartmentalized case are shown in parentheses. The distributions are based on $10^4$ 
	trajectories. The red line shows the empirical distribution of the maximal phenotypic value in the first generation (a sample of size $N$ drawn 
	from the distribution $\rho_0$).
}
\label{sim-classic}
\end{figure}

To distinguish these effects, we performed the same individual-based simulation but without random compartmentalization, where each individual 
directly acquired the fitness equal to its phenotype. Focusing on the cases that corresponded to $M = 100$ on Figure~\ref{simulationsMN-all}, where the 
compartmentalized case showed fixation of a single phenotype in individual trajectories, we ran the noncompartmentalized simulations with populations 
sizes $N = 200$ and $N = 20$ to be compared with $\lambda = 2$ and $\lambda = 0.2$, respectively. The result of these simulations is shown on 
Figure~\ref{sim-classic}, where again a mean distribution of the fixed phenotypes are shown averaged for $10^4$ individual trajectories. One can see 
that the main part of the variability in the noncompartmentalized case comes from the sampling of the initial population, as the distribution of the 
fixed phenotypes closely matches the (numerically computed) distribution of the maximum phenotypic value in a sample of size $N$ from $\rho_0$. The 
final distribution, however, is broadened and slightly shifted to low values of fitness. We attribute this to the genetic drift. Note how the 
compartmentalization does not affect this process in case of small $\lambda$ but significantly worsens the final distribution in case of large 
$\lambda$---the effect that, in contrast to the regular genetic drift, becomes \emph{stronger} with the \emph{increase} of the population size. This 
effect may be a manifestation of the increase in the correlations between the fitness of individuals. Another possible explanation could be the 
weakening of the deterministic part of the selection process (predicted by equation (\ref{update-rho}) for large values of $\lambda$). This could give 
more chances for stronger mutants to be lost before their frequency grows high enough. It is difficult to separate these effects from this kind of data 
so we performed an additional numerical experiment designed specially for this purpose.

\subsection{Stochasticity in a bivariate population under linear and nonlinear selection}

\begin{figure}[t]
\centering
\includegraphics[scale=0.7]{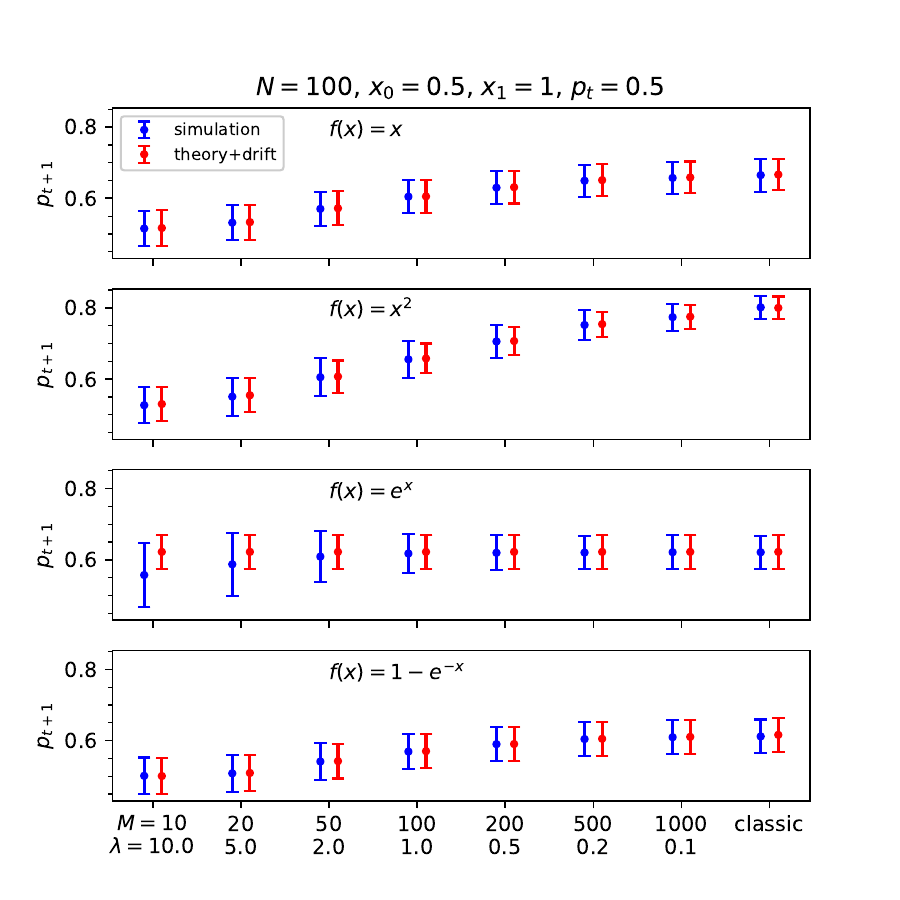}
\caption{
	Mean gene frequency $p$ of the more active mutant with phenotype $x_1$ (resident phenotype $x_0$) in a small population of size $N$ and its 
	standard deviation after one selection cycle with and without compartmentalization ($M$ compartments) and with various selection functions $f$. 
	Statistical properties are computed on a sample of $10^5$ individual numerical experiments (in blue). These values are compared to the prediction 
	of the deterministic infinite population theory and with the standard deviation given by a simple genetic drift due to sampling (in red). The case 
	marked ``classic'' corresponds to the classical selection without any compartmentalization, where the selection function directly defines the 
	fitness.
}
\label{fig-drift}
\end{figure}

To distinguish the stochastic effect of compartmentalization from the deterministic reduction of the selection pressure predicted by 
(\ref{update-rho}), we performed individual-based simulations on a population of much simpler structure (a bivariate case). More specifically, we took 
the initial distribution with density $\rho = p\delta_{x_1} + (1-p)\delta_{x_0}$, where we assume $x_0$ to be the wild type phenotype, $x_1$ to be an 
invading mutant, and $p$ to be the gene frequency if this mutant.

The case of a bivariate population makes the comparison simpler, as in the classical situation, the mean change of the frequency in one cycle of 
selection is given by the deterministic part of selection obtained in the infinite population limit, while the variance of this change results from the 
sampling stochasticity. To quantify the additional stochastic effect of the compartmentalization, we compared the mean frequency change to the value 
predicted by (\ref{update-rho}) for the corresponding value of $\lambda$ and its variance with the sampling variance (classical genetic drift) for the 
corresponding population size. The results are depicted on Figure~\ref{fig-drift}. Interestingly, the combination of the deterministic theory developed 
in our work with the sampling variance due to drawing of the new population (the standard genetic drift) completely explains the observed stochasticity 
in the compartmentalized selection with the linear selection function.

We also performed the same simulations for nonlinear selection functions theoretically studied in this article: a quadratic selection, an exponential 
one, and $f(x) = 1 - e^{-x}$ (see (\ref{quadratic}) and (\ref{exponential}) with the following remark). The latter function, being a sum of two 
exponentials, $e^{0x}$ and $e^{-x}$, and thus completely described by (\ref{exponential}), is an example of a saturating selection: it monotonously 
increases but is limited from above by 1. The results of these simulations are shown on Figure~\ref{fig-drift}, too. Although the quadratic case, and 
to even lesser extent the saturating selection, shows a wider variance of the gene frequency in one cycle than predicted by the genetic drift, the 
effect is very weak. Only the exponential case strongly demonstrates the effect of the compartmentalization when the number of compartments decreases 
($\lambda$ increases). Interestingly, the deviation of the noise from the genetic drift happens along with the deviation of the mean from the 
prediction based on the infinite population theory.

We can conclude that the deterministic theory with genetic drift works well even in very small populations at least for not very strongly growing 
selection functions. However, the behaviour of the exponential selection clearly demonstrates some new effect that deserves a separate investigation.

\section*{Discussion and Conclusion}

In this work we considered a model of the group selection with no intragroup competition, where grouping happens by a random encapsulation of the 
individuals in a collection of compartments at the beginning of each replication cycle. This study was initially motivated by experiments from the 
field of artificial molecular evolution. In particular, compartmentalized \emph{in vitro} directed evolution protocols provide a very pure 
implementation of our model, and may in the future be used to test its predictions. 

We note however that our framework can also apply to some situations of natural evolution. Viruses, for example, exist as highly polymorphic 
populations, or quasi-species \cite{Domingo2001}, which randomly infect individual cells \cite{Manrubia2006} to carry their replication cycle. These 
cells thus serve as genotype-phenotype linkage-maintaining compartments. The possibility of co-infection of a given cells by multiple, possibly 
non-clonal viral particles is largely acknowledged, and generally parametrized by the averaged ``multiplicity of infection'' 
\cite{Gonzalez-Jara2009,Frank2001,Novella2004}. Additionally, inside the cells, the genetic and viral replication is carried according to the 
phenotypic activity of gene products from the viral genome \cite{Manrubia2006}. This mode of reproduction naturally involves template 
indiscriminateness, and is thus accurately described by our model.

In the context of the origin of life, it is also assumed that ancestral replicators were RNA molecules functioning as universal RNA-depended 
RNA-polymerases with activity \emph{in trans}. These primordial replicators possibly used naturally formed vesicles, coacervates, pores in regolith, or 
other microcompartmetalized niches. This implies a random compartmentalization---hence frequency-dependent---stage in their replication cycle. In a 
similar way, modern parasites may also fall in our conceptual scheme, because multiple unrelated individuals may infect a single host and modify their 
joint virulence/survival as a result of interaction \cite{Davies2002,Fried1981,Fried1990}.

Finally, the process is similar to selection in polyploid organisms under random mating. The analogy becomes exact if one considers zygotes to be 
compartments and the haploid genomes to be the individuals under selection. There are two particularities though: the fixed group size in the 
polyploids (the ploidy of the organism) and the presence of the genetic cross-over, not considered in our model. Thus, the problem considered in the 
present article can be viewed as a generalization of natural selection theory to organisms with variable ploidy, in the absence of recombination.

Our model includes a number of simplifications: no selection inside compartments (pure group selection), infinitely large populations, no overlap 
between generations, completely random occupation of compartments by individuals, absence of mutations, deterministic phenotypic expression, and (in 
most parts) a single additive phenotypic trait. Many of these assumptions are relevant to the initial motivation of the work (\emph{in vitro} 
micro-compartmentalized evolution).

The most restrictive condition is the demand of the additivity of the phenotype in a compartment. In the context of directed evolution of enzymes, when 
the phenotypic trait is the enzymatic activity, this hypothesis neglects any effect of the activity saturation. Another case of failure of the 
additivity is brought by the independence of the total phenotype in a compartment on the number of individuals in it \emph{per se}. This situation is 
encountered, for example, when individuals are bacteria that contain a plasmid (plasmids) with the gene and the appropriate protein product, like 
\cite{Ghadessy2001}. If the content of a bacterial cell has an effect on critical reactions during the reproduction phase, the additivity of 
per-individual activities fails.

The deterministic infinite population limit is not very restrictive, as we demonstrated by numerical simulations in small finite populations. First, 
typical population sizes in some potential domains of application for our theory (high throughput \emph{in vitro} evolution, viral evolution) are very 
large. Second, when the population is small, the developed theory well predicts the average result of selection. The stochastic component, in turn, can 
be almost completely attributed to and taken in to account by the standard genetic drift, despite the compartmentalization. This situation is similar 
to the standard mathematical treatment of the selection process in diploid organisms, where the Hardy-Weinberg equilibrium followed by the 
deterministic selection is coupled to the haploid genetic drift \cite{Ewens2004}. In the light of the analogy between a compartmentalized population 
and a population of organisms with variable ploidy, the infinite population limit is equivalent to the Hardy-Weinberg equilibrium assumption. This 
approximation fails, however, if the selection function grows too fast, like in the exponential selection.

The initial motivating question for our work was the influence of the average number of individuals per compartment ($\lambda$) on the selection 
dynamics. We answered this question using distribution theoretical tools. The main result is a quantitative characterization of the co-encapsulation 
effect in the form of a discrete-time dynamical equation for the phenotypic distribution that uses $\lambda$ as the main parameter. Using this 
dynamical description, we thoroughly analyzed the simplest case: additive phenotype with linear selection function under Poisson partitioning in the 
compartments. The main conclusion is that the selection process, contrary to a common belief in the field of directed evolution, is still effective 
even in the context of mixing many individuals with different phenotypes in the same compartment. In quantitative terms, random co-encapsulation slows 
down the selection dynamics (the rate of the mean fitness change) by a factor that decays approximately as $1/\lambda$. With ten individuals per 
compartment in average, the selection is roughly ten-fold slower than that for the extremely diluted case, when any individual is alone in its 
compartment. In all cases, however, independent of $\lambda$, the most active mutant is still eventually selected for.

We also developed a framework to treat more complicated cases like nonlinear fitness, nonadditive phenotype, and multiple trait phenotype. For some 
special cases (polynomial and exponential selection), we obtained dynamical equations in closed form with an explicit effect of $\lambda$. The striking 
observation here is that different nonlinearities have very different resistance to random co-encapsulation. In particular, the exponential selection 
is completely immune to this effect (in the infinite population approximation). This is especially interesting as some important cases of evolving 
replicators can be approximated by this regime, e. g. when the phenotypic trait directly controls the catalytic rate of the replication.

Although this is not the first theoretical work on compartmentalized selection, we greatly generalized previous research, which only considered 
populations of few phenotypes. Our model is suitable for the treatment of arbitrary phenotypic distributions, discrete and continuous alike, and 
arbitrary laws of distribution of individuals among the compartments. Thanks to the advanced mathematical tools applied in this work, the computations 
became more transparent and simple in comparison to the traditional combinatorial approach, which gets very cumbersome already for a quadratic 
selection function, even in the case of a simple bivariate population.

Finally, one practical implication of this study concerns the design of \emph{in vitro} evolution experiments. While many practitioners tend to 
empirically select very small values of $\lambda$, we show that higher library concentration will not drastically affect the performance of the 
experiment. The corresponding increase of the throughput is, however, paid for by the decrease of the selection pressure, and more rounds may become 
necessary to fix the best genotypes. A more comprehensive assessment of the cost-benefit balance for directed evolution experiments will be published 
elsewhere.

\section*{Acknowledgements}
The authors are grateful to David Lacoste and Luca Peliti for stimulating discussions and especially to Ken Sekimoto for numerous discussions and for 
critically reading the manuscript. We also would like to thank an anonymous reviewer for pointing out a noncritical but unpleasant mathematical 
mistake in the manuscript. This work was supported by European Research Council (ERC, Consolidator grant 647275 PROFF).

{\appendix

\section*{Appendices}

\setcounter{section}{1}

\subsection{Random variables and generalized functions}
\label{generalized-functions}

This appendix does not contain new results. Its purpose is to be an introduction to the framework used in the article.

A random variables $X$ on $\mathbb R$ is associated with a positive Radon measure $\mu_X$, which can be understood as a non-negative generalized 
function (linear continuous functional) $\rho_X$ defined on the space of continuous functions with compact supports $C_c$ endowed with the appropriate 
topology (see, for example, \cite{Schwartz1}). This topology is conventionally induced by the following convergence rule: we say $\phi_n \to \phi$ in 
$C_c$, if there is a compact $K$ such that $\forall n$ $\supp \phi_n \subset K$ and $\phi_n \to \phi$ homogeneously. We will denote the topological 
linear space of generalized functions on $C_c$ endowed with the weak topology as $C_c'$, so the topology on $C_c'$ is induced by the convergence 
$\rho_n \to \rho \Leftrightarrow \forall \phi \in C_c$ $\langle \rho_n, \phi \rangle \to \langle \rho, \phi \rangle$. Such convergence is equivalent to 
the weak convergence of the corresponding measures in the measure-theoretical language. As this is the only convergence of generalized functions that 
we consider in the article, the use of the sign $\to$ to denote it does not bring any confusion.

A generalized function $\rho$ is called non-negative, if for 
any non-negative $\phi \in C_c$ (that is $\forall x \in \mathbb R$ $\phi(x) \geqslant 0$) we have $\langle \rho, \phi \rangle \geqslant 0$. We will 
denote the subset of non-negative generalized functions with the subset topology as $C'_{c+}$.

It is possible to extend the action of $\rho$ from functions in $C_c$ to any indicator function of a Borel set $\chi_B$, $B \in \mathfrak B$, where 
$\mathfrak B$ is the Borel $\sigma$-algebra. This is done using the so called upper and lower value of $\rho$ on an indicator of a set (in 
\cite{Schwartz1} this corresponds to the upper and the lower measure of a set). The upper value of $\chi_A$ is defined by
	\begin{equation}
	\langle \hat\rho, \chi_A \rangle = \inf_{\scriptsize \begin{matrix}U \\ A \subset U \end{matrix}}
	\sup_{\scriptsize \begin{matrix}\supp \phi \subset U \\ 0 \leqslant \phi \leqslant 1 \end{matrix}}
	\langle \rho, \phi \rangle,
	\end{equation}

\noindent where $U$ are open sets in the standard topology on $\mathbb R$, and $\phi \in C_c$. The lower value is defined as
	\begin{equation}
	\langle \check\rho, \chi_A \rangle = \sup_{K \subset A}
	\inf_{\scriptsize \begin{matrix}\phi \geqslant 0 \\ \phi|_K = 1\end{matrix}}
	\langle \rho, \phi \rangle,
	\end{equation}

\noindent where $K$ are compact. If, for \emph{compactly supported non-negative} $\rho$, $\langle \hat\rho, \chi_A \rangle = \langle \check\rho, \chi_A 
\rangle$, the set $A$ is called $\rho$-measurable and the value $\langle \rho, \chi_A \rangle$ is defined as this common value and is called the 
measure of $A$. It can be proven that for any compactly supported non-negative generalized function $\rho$ any Borel set is $\rho$-measurable (it is 
also $\sigma$-regular, see\cite{Schwartz1}). In particular, the whole $\mathbb R$ is measurable, and in fact, for $\rho$ to be a (generalized) 
probability density functions, we require $\langle \rho, \chi_{\mathbb R} \rangle \deq \langle \rho, 1 \rangle = 1$. For \emph{non-negative} $\rho$ 
\emph{that have} $\langle \rho, 1 \rangle < \infty$, any set $A$ such that $\langle \hat \rho, \chi_A \rangle = \langle \check \rho, \chi_A \rangle$ is 
called $\rho$-measurable, whether $\rho$ is compactly supported or not. Note that for a non-negative $\rho$ and a continuous $f$, $\langle \rho, 
f\rangle$ is well defined as $\langle f\rho,1\rangle$ and one does not need any further development of the theory.

Any point-set is measurable, too, and we can compute a measure of a point $x$ as
	\begin{equation}
	\langle \rho, \chi_{\{x\}} \rangle = \inf_{\scriptsize \begin{matrix}\phi \geqslant 0 \\ \phi(x) = 1\end{matrix}} \langle \rho, \phi \rangle.
	\end{equation}

From now on we will call non-negative generalized functions with $\langle \rho , 1 \rangle = 1$ (probability) \emph{densities} (of random variables). 
The meaning of a density $\rho_\xi$ of a random variable $\xi$ is that the probability to find the value of $\xi$ in a set $A$ is equal to $\langle 
\rho_\xi, \chi_A \rangle$. A (cumulative) distribution function of $\xi$ is the function $F_\xi \colon x \mapsto \langle \rho_\xi, \chi_{(-\infty,x]} 
\rangle$. The density is the (generalized) derivative of its own distribution function. The mathematical expectation (the mean) of the random variable 
is then computed as $\bar x = \langle \rho, x \rangle$.

Any density $\rho$ can be uniquely decomposed into a sum
	\begin{equation}
	\rho = \rho_a + \rho_p + \rho_r,
	\end{equation}

\noindent where $\rho_a$ is the regular part, $\rho_p$ is the point-mass part, and $\rho_r$ is the residual singular part.

The regular part $\rho_a$ is a regular generalized function, i.e. its action on any $\phi \in C_c$ can be represented by $\langle \rho_a , \phi 
\rangle = \intl{}{} f\phi \,dx$ for a unique $f \in L^1_\mathrm{loc}$ (the integration is in the sens of Lebesgue). The regular part is also called the 
absolutely continuous part (hence the notation $\rho_a$). A density that has only this part is also called absolutely continuous. Any point has a zero 
measure in respect to an absolutely continuous density. It is convenient to identify (a representative of) $f$ with $\rho_a$ and to write $\rho_a$ 
instead of $f$.

The point-mass part $\rho_p$ is an at most countable sum of $\delta$-functions
	\begin{equation}
	\rho_p = \sum_{n \in \mathbb N} a_n \delta_{x_n}, \quad a_n \geqslant 0, \quad n \neq m \Rightarrow x_n \neq x_m,\quad \sum_{n \in \mathbb N} a_n 
	\leqslant 1.
	\end{equation}

\noindent It follows that $\forall x \neq x_n$ $\langle \rho_p , \chi_{\{x\}} \rangle = 0$, and $\forall x_n$ $\langle \rho_p , \chi_{\{x_n\}} \rangle 
= a_n$. Physically speaking, $\rho_p$ represents all fitness values that are present in macroscopic quantities in the population.

Finally, the residual singular part $\rho_r$ is characterized by the zero (Lebesgue) measure of its support and, at the same time, $\forall x$ 
$\langle \rho_r, \chi_{\{x\}} \rangle = 0$. Its support is Cantor set-like and its distribution function is Cantor function-like.

The regular and the residual singular parts form together the continuous part $\rho_c = \rho_a + \rho_r$. The sum of the point-mass and the residual 
singular parts is the singular part $\rho_s = \rho_p + \rho_r$. It is tempting to disregard $\rho_r$ as unphysical. However, it may turn to be a good 
tool to model libraries obtained by a random mutagenesis from a single mutant in a very rugged fitness landscape. It might be possibly a good 
approximation to a library generated on a smooth landscape but by an error-prone PCR with large number of cycles.

\subsection{Continuity of operator $A$ and of the operators that generate $\sigma_x$ and $\sigma$ for the linear fitness case}
\label{app-continuity}

Let $\mathbb P$ be the space of probability densities with, so $\mathbb P = \{\rho \in C_{c+}'\,|\, \langle\rho,1 \rangle\ = 1\}$. Let $\mathbb P_p$ be 
the space of finite point-mass probability densities, so $\mathbb P_p = \{\rho \in \mathbb P \,|\, \exists\, n \in \mathbb N \,, \rho = 
\sum\limits_{k=0}^n a_k \delta_{x_k}\}$. Let us fix some very large positive number $\mathcal L$. Let us denote $\mathcal I \deq [0,\mathcal L]$. Let 
$\mathbb P^\mathcal{I}$ be the space of probability densities concentrated in $\mathcal I$, so $\mathbb P^\mathcal{I} = \{\rho \in \mathbb P \,|\, 
\supp \rho \subset \mathcal I\}$. Let $\mathbb P^\mathcal{I}_p$ be the space of finite point-mass probability densities from $\mathbb P^\mathcal{I}$, 
so $\mathbb P^\mathcal{I}_p = \mathbb P^\mathcal{I} \cap \mathbb P_p$. Let the topologies on all these spaces be inherited from $C_c'$, subspaces of 
which all of them are.

In Section~\ref{selection-arbitrary}, we derived the update operator $\rho_{t+1} = A(\rho_t)$ given by the formula (\ref{update-rho}) for densities 
from $\mathbb P_p$. This formula is generalizable on any generalized function $\rho$ from $\mathbb P$ with nonzero finite mean $\langle \rho, x 
\rangle$. Let us denote $N \deq \{\rho \in \mathbb P\,|\, \langle \rho,x\rangle = 0\}$. Let us also denote $N_p \deq N \cap \mathbb P_p$. If $N_p$ was 
nowhere dense in $\mathbb P_p$, if, in addition, $\mathbb P_p$ was dense in $\mathbb P$ and (\ref{update-rho}) happened to be continuous both in 
$\mathbb P_p\setminus N_p$ and $\mathbb P\setminus N$, we could take this expression as an extension by continuity of $A$ from $\mathbb P_p$ to 
$\mathbb P$. Unfortunately, this assertion is not true.

\begin{proposition}
For any $\lambda > 0$, the operator $A\colon \mathbb P_p \setminus N_p \to \mathbb P_p$ defined by (\ref{update-rho}) is nowhere continuous.
\label{discontinuous}
\end{proposition}

\begin{proof}

Take any $\rho \in \mathbb P_p \setminus N_p$. By the definition of $\mathbb P_p$, we have $0 < |\langle \rho,x\rangle| < \infty$. Consider the 
sequence
	\begin{equation}
	\rho_n = \dfrac{n}{n+1}\rho + \dfrac{1}{n+1}\delta_{n^2}.
	\end{equation}

\noindent For any large enough $n$, $\rho_n \in \mathbb P_p \setminus N_p$. Furthermore, $\rho_n \to \rho$ in the topology of $\mathbb P_p \setminus 
N_p$. Indeed, for any test function $\phi \in C_c$ there is a number $n_0$ such that for any $n > n_0$ the point $n^2$ does not belong to the support 
of $\phi$. Therefore, for $n > n_0$ we have
	\begin{equation}
	\langle \rho_n, \phi\rangle = \dfrac{n}{n+1}\langle \rho,\phi\rangle \to \langle \rho,\phi \rangle.
	\end{equation}

\noindent From the other hand, $\rho_n$ does not converge to $\rho$ in mean. Indeed,
	\begin{equation}
	\langle \rho_n,x\rangle = \dfrac{n}{n+1}\langle \rho,x\rangle + \dfrac{n^2}{n+1} \to +\infty.
	\end{equation}

\noindent Therefore,
	\begin{equation}
	A(\rho_n) \to \Big(1 - g(\lambda)\Big) \rho \neq A(\rho) = \left(1 - g(\lambda) + g(\lambda)\frac{x}{\langle \rho,x\rangle}\right)\rho.
	\end{equation}
\end{proof}

This deplorable fact, however, can be remedied by the restriction of the space of the considered probability densities to $\mathbb P^\mathcal{I}$ 
for some $\mathcal I$. This restriction has purely technical meaning and does not reflect any physical constraints. Nevertheless, some justification 
comes from the fact that there is a universal upper bound on the activity of any enzyme of the considered class of enzymes. This upper bound is 
reflected by the number $\mathcal L$ that defines $\mathcal I$. Note that $\mathbb P^\mathcal{I} \cap N = \mathbb P^\mathcal{I}_p \cap N = 
\{\delta_0\}$.

\begin{theorem}
For any $\mathcal I$, the operator $A\colon \mathbb P^\mathcal{I} \setminus \{\delta_0\} \to \mathbb P^\mathcal{I}$ defined by the formula 
(\ref{update-rho}) is continuous.
\label{theorem-continuous}
\end{theorem}

\begin{proof}

Let $\rho_n$ be some sequence from $\mathbb P^\mathcal{I}\setminus \{\delta_0\}$ that converges to some $\rho \in \mathbb P^\mathcal{I}\setminus 
\{\delta_0\}$. Then
	\begin{equation}
	|\langle \rho_n,x\rangle - \langle \rho,x\rangle| = |\langle \rho_n - \rho, x \rangle| = \mathcal |\langle \rho_n - \rho, \eta \rangle |
	\to 0,
	\label{mean-convergence}
	\end{equation}

\noindent where	$\eta$ is some function from $C_c$ such that $\forall x \in \mathcal I$ $\eta(x) = x$.

It follows that
	\begin{equation}
	\left(1-g(\lambda) + g(\lambda)\frac{x}{\langle \rho_n, x\rangle}\right) \to \left(1-g(\lambda) + g(\lambda)\frac{x}{\langle \rho, x\rangle}\right)
	\label{functions-phi}
	\end{equation}

\noindent point-wise on $\mathcal I$. As all the involved functions are also continuous on the compact set $\mathcal I$, the convergence is uniform.

Let us denote the function on the right-hand side of (\ref{functions-phi}) as $\tilde\phi$ and the functions on the left-hand side as $\tilde\phi_n$ 
(for each $n$, $\tilde\phi_n$ corresponds to the function generated by $\rho_n$). Then it is always possible to find functions $\phi,\phi_n \in C_c$ 
such that $\phi_n \to \phi$ in $C_c$ and $\left.\phi\right|_\mathcal{I} = \left.\tilde\phi\right|_\mathcal{I}$, $\left.\phi_n\right|_\mathcal{I} = 
\left.\tilde\phi_n\right|_\mathcal{I}$ and, therefore, $\tilde\phi_n\rho = \phi_n\rho$, $\tilde\phi\rho = \phi\rho$.

As $A(\rho_n) = \phi_n\rho_n$ and $A(\rho) = \phi\rho$, what is left to be proven is that given $\rho_n \to \rho$ and $\phi_n \to \phi$ we have $\phi_n 
\rho_n \to \phi \rho$. First notice that $\sup |\phi_n\psi - \phi\psi| \to 0$ for any $\psi \in C_c$ (from which it follows that $\phi_n \psi \to \phi 
\psi$ in $C_c$). Therefore, for any $\epsilon > 0$ there is $n_0$ such that for any $n > n_0$ we have $\sup |\phi_n\psi - \phi\psi| < \epsilon$. We 
have, therefore, for any $n > n_0$
	\begin{multline}
	|\langle\phi_n\rho_n - \phi\rho,\psi\rangle| \leqslant |\langle\rho_n - \rho, \phi\psi\rangle| + |\langle\rho_n,\phi_n\psi - \phi\psi\rangle| 
	\leqslant\\
	\leqslant
	|\langle\rho_n - \rho,\phi\psi\rangle| + \sup|\phi_n\psi - \phi\psi| < |\langle\rho_n - \rho,\phi\psi\rangle| + \epsilon.
	\end{multline}

But $\rho_n \to \rho$ means that there is $n_1 > n_0$ such that for any $n > n_1$ we have $|\langle\rho_n - \rho,\phi\psi\rangle| < \epsilon$. 
Therefore, for any $\epsilon$ and for any $\psi$ we have $|\langle A(\rho_n) - A(\rho),\psi\rangle| < 2\epsilon$ starting from $n_1$.

\end{proof}

As it is seen from the proof, the statement of the theorem stays correct, if we replace $\mathcal I$ with any compact set $K$ with nonempty interior 
and if we replace $\mathbb P^\mathcal{I}\setminus \{\delta_0\}$ with $\mathbb P^K \setminus (\mathbb P^K \cap N)$, where $\mathbb P^K = \{\rho \in 
\mathbb P\,|\,\supp \rho \subset K\}$. Furthermore, the set $\mathbb P^K_p \cap N$ is nowhere dense in $\mathbb P^K_p$ and the set $\mathbb P^K \cap N$ 
is nowhere dense in $\mathbb P^K$.

The only thing that is left to be proven is that the space of finite discrete densities is dense in the space of general densities.

\begin{theorem}
For any $\mathcal I$, the space $\mathbb P^\mathcal{I}_p$ is dense in $\mathbb P^\mathcal{I}$ as its subset.
\label{theorem-dense}
\end{theorem}

\begin{proof}

First let us prove that the space $\mathbb P^\mathcal{I} \cap C_c$ is dense in $\mathbb P^\mathcal{I}$, where $C_c$ is understood as being naturally 
embedded into $C_c'$. That is, any density from $\mathbb P^\mathcal{I}$ can be approximated by a sequence of densities from $C_c$ with the supports in 
$\mathcal I$.

Consider some $\omega \in C_{c+}$ such that $\forall x$ $\omega(-x) = \omega(x)$ and $\intl{}{} \omega(x)\,dx = 1$. Let us denote $r = \diam \supp \omega$ 
and $\omega_n(x) = \dfrac{1}{n}\omega\left(\dfrac{x}{n}\right)$. Let us also consider the sequence of mappings $F_n: \mathbb R \to \mathbb R$, $x 
\mapsto \dfrac{\mathcal L }{\mathcal L + 2\frac{r}{n}}\left(x + \dfrac{r}{n}\right)$. For each $n$, $F_n$ bijectively maps the interval 
$\left[-\dfrac{r}{n}, \mathcal L + \dfrac{r}{n}\right]$ to the interval $[0,\mathcal L] = \mathcal I$.

For any generalized function $\rho \in \mathbb P^\mathcal{I}$ take the sequence of $\psi_n = (F_n)_\star (\rho * \omega_n) \in \mathbb P^\mathcal{I} 
\cap C_c$. Since $\omega_n \to \delta_0$ in $\mathbb P$, as $n \to \infty$, we have $\rho * \omega_n \to \rho$ in $\mathbb P$. Let us prove that 
$\psi_n \to \rho$.

For any $\phi \in C_c$ we have $\langle (F_n)_\star (\rho * \omega_n), \phi \rangle = \langle \rho * \omega_n, \phi\circ F_n\rangle$. It is not 
difficult to show that $\phi \circ F_n \to \phi$ in $C_c$. The point $x_0 = \mathcal L/2$ is the stationary point for all $F_n$. All $F_n$ are affine 
and contracting with the contraction coefficient $\mathcal L/(\mathcal L + 2r/n)$. Their inverses $F_n^{-1}$ are expanding with the expansion 
coefficient $\kappa_n \deq 1 + 2r/(n\mathcal L)$. Note that $\forall n$, $\kappa \deq \kappa_1 \geqslant \kappa_n$. Let $\Delta \deq \max\Big(|x_0 - 
\inf\supp\phi|, |x_0 - \sup\supp\phi|\Big)$ and $K \deq [x_0 - \kappa\Delta, x_0 + \kappa\Delta]$. Then $\forall n > 0$, $\supp \phi \circ F_n \subset 
K$ and $\supp \phi \subset K$. As $F_n \to \operatorname{Id}_{\mathbb R}$ pointwise on $\mathbb R$ and $F_n$ are continuous, this convergence is 
uniform on $K$. That is $\sup\limits_{x \in K} |F_n(x) - x| \to 0$. As $\phi$ is continuous and compactly supported, it is also uniformly continuous. 
Therefore,
	\begin{equation}
	\sup_{x \in \mathbb R}\Big|\phi\Big(F_n(x)\Big) - \phi(x)\Big| \to 0.
	\end{equation}

\noindent But together with $\supp \phi \circ F_n \subset K$ this proves that $\phi \circ F_n \to \phi$ in $C_c$.

We have $\rho * \omega_n \to \rho$ and $\phi \circ F_n \to \phi$. Using the same reasoning as in the proof of Theorem~\ref{theorem-continuous}, we 
conclude that $\langle \psi_n, \phi\rangle = \langle \rho * \omega_n, \phi \circ F_n\rangle \to \langle \rho,\phi \rangle$, and thus, $\psi_n \to 
\rho$.

From the other hand, any $\psi \in \mathbb P^\mathcal{I} \cap C_c$ can be approximated by a sequence from $\mathbb P^\mathcal{I}_p$. Indeed, we can 
select some sequence of conventionally ordered Darboux partitions $\{\Delta_n\}$ of some interval $[a,b]$ that contains $\supp \psi$ (one can take tho 
whole $\mathcal I$) with the graininess of the partitions going to 0 with $n \to \infty$, where $\Delta_n = \{x^{(n)}_k\}$, $k \in \{0, 1, \ldots , 
K_n\}$, $x^{(n)}_k < x^{(n)}_{k+1}$, $x^{(n)}_0 = a$, $x^{(n)}_{K_n} = b$, $\Delta x^{(n)}_k = x^{(n)}_{k+1}-x^{(n)}_k$, and $\max\limits_{k < K_n} 
\Delta x^{(n)}_k \to 0$, when $n \to \infty$. Then we can take the sequence of $\rho_n \in \mathbb P^\mathcal{I}_p$ of the following form
	\begin{equation}
	\rho_n = \sum_{k = 0}^{K_n - 1}\Delta x^{(n)}_k\psi(\xi^{(n)}_k)\delta_{\xi^{(n)}_k},
	\end{equation}

\noindent where $\xi^{(n)}_k \in [x^{(n)}_k, x^{(n)}_{k+1}]$ such that
	\begin{equation}
	\psi(\xi^{(n)}_k) \Delta x^{(n)}_k = \intl{x^{(n)}_k}{x^{(n)}_{k+1}} \psi(x)\,dx.
	\end{equation}

\noindent Such $\xi^{(n)}_k$ always exist by the mean value theorem. Their role is to enforce $\langle \rho_n,1\rangle = 1$.

Then for any $\phi \in C_c$ we have
	\begin{multline}
	\langle \rho_n, \phi \rangle = \langle \sum_{k = 0}^{K_n - 1}
	\Delta x^{(n)}_k\psi(\xi^{(n)}_k)\delta_{\xi^{(n)}_k}, \phi \rangle =\\
	=\sum_{k = 0}^{K_n - 1}\Delta x^{(n)}_k\psi(\xi^{(n)}_k)\phi(\xi^{(n)}_k) \to
	\intl{a}{b}\psi(x)\phi(x)\,dx = \langle \psi,\phi \rangle,
	\end{multline}

\noindent as both $\psi$ and $\phi$ are continuous, and thus, $\psi\phi$ is Riemann integrable. As $\mathbb P^\mathcal{I} \cap C_c$ is dense in 
$\mathbb P^\mathcal{I}$, it follows that $\mathbb P^\mathcal{I}_p$ is dense in $\mathbb P^\mathcal{I}$.

\end{proof}

Note that the proof of the theorem is easily extended to probability densities on $\mathbb R^n$, the case important for a multitrait selection 
considered in Section~\ref{multitrait}. The only difference is that in this case Riemann sums are built on the base of Jordan partitions of a 
Jordan-measurable set (a simple rectangular parallelepiped is enough) that contains $\supp \psi$.

We now will prove that the operators that generate $\sigma_x$ and $\sigma$ from $\rho$ defined by formulas (\ref{sigma-x}) and (\ref{sigma}) are 
continuous, too, and that they can be thus extended to any probability density. It means that they can be used independently of $A$, if the situation 
demands it. Let us denote these operators from $\mathbb P$ to itself as $\Sigma_x$ and $\Sigma$. Let us also denote $\sigma^\rho_x \deq \Sigma_x(\rho)$ 
and $\sigma^\rho \deq \Sigma(\rho)$ to be able to distinguish fitness distributions generated by different phenotypic distributions.

\begin{theorem}

For any $x \in \mathbb R$, the operators $\Sigma_x, \Sigma \colon \mathbb P \to \mathbb P$ are continuous.
\label{sigma-continuous}

\end{theorem}

\begin{proof}

We will prove the theorem only for $\Sigma_x$. The proof for $\Sigma$ is analogous. The proof is essentially based on the absolute convergence of all 
the involved numerical series.

Let us choose any sequence of $\rho_n \in \mathbb P$ that converges to some $\rho \in \mathbb P$ in $\mathbb P$. We need to prove that 
$\Sigma_x(\rho_n) \to \Sigma_x(\rho)$. Let us choose some $\phi \in C_c$. The value of $\langle \Sigma_x(\rho), \phi \rangle$ is equal to
	\begin{equation}
	\langle \Sigma_x(\rho), \phi \rangle = \sum_{k=0}^\infty P_k \langle \delta_x * \rho^{*k}, \phi \circ h_{k+1}\rangle,
	\end{equation}

\noindent where $P_k = e^{-\lambda}\lambda^k/k!$ and $h_k\colon x \mapsto x/k$.

First note that $\sup\limits_x |\phi \circ h_k (x)| \leqslant \sup\limits_x |\phi(x)|$. Let us denote $\Phi \deq \sup\limits_{x} |\phi(x)|$. Then the 
following estimate is correct
	\begin{multline}
	|\langle \delta_x * \rho_n^{*k} - \delta_x * \rho^{*k}, \phi\circ h_{k+1} \rangle| 
	\leqslant
	|\langle \delta_x * \rho_n^{*k},\phi\circ h_{k+1}\rangle| + |\langle\delta_x * \rho^{*k}, \phi\circ h_{k+1} \rangle| 
	\leqslant  \\ \leqslant 
	\Phi\Big(\langle \delta_x * \rho_n^{*k},1\rangle + \langle\delta_x * \rho^{*k}, 1\rangle\Big) = 2\Phi.
	\label{Phi-upper-bound}
	\end{multline}
This, in turn, implies
	\begin{equation}
	|\langle\Sigma_x(\rho_n) - \Sigma_x(\rho),\phi\rangle| \leqslant 
	\sum_{k=0}^\infty P_k|\langle\delta_x*\rho_n^{*k}-\delta_x*\rho^{*k},\phi\circ h_{k+1}\rangle|
	\leqslant 2\Phi\sum_{k=0}^\infty P_k = 2\Phi < \infty.
	\end{equation}
Therefore, for any $\epsilon > 0$ there is $k_0$ such that $\sum\limits_{k=k_0}^\infty 2\Phi P_k < \epsilon$, and, thus, by (\ref{Phi-upper-bound}), 
for any $n$
	\begin{equation}
	\sum_{k=k_0}^\infty P_k|\langle\delta_x*\rho_n^{*k}-\delta_x*\rho^{*k},\phi\circ h_{k+1}\rangle| < \epsilon.
	\end{equation}

From the other hand, for any $k$, $\rho_n \to \rho$ implies $\delta_x*\rho_n^{*k} \to \delta_x*\rho^{*k}$. Therefore, for any $\epsilon > 0$ and any 
$k_0$ there is $n_0$ such that for any $n > n_0$
	\begin{equation}
	\sum_{k=0}^{k_0-1} P_k|\langle\delta_x*\rho_n^{*k}-\delta_x*\rho^{*k},\phi\circ h_{k+1}\rangle| < \epsilon.
	\end{equation}

Using the following intuitive logical formula
	\begin{equation}
	\Big(\forall x\exists y\forall \z\, A(x,y,\z)\Big) \wedge \Big(\forall x\forall y \exists \z\, B(x,y,\z) \Big)
	\Rightarrow \forall x\exists y\exists \z\ \Big(A(x,y,\z) \wedge B(x,y,\z)\Big),
	\label{lemma-logic}
	\end{equation}

\noindent where $A$ and $B$ are some propositional functions in three variables, we conclude that for any $\epsilon > 0$ there exists $n_0$ such that 
for any $n > n_0$
	\begin{equation}
	\sum_{k=0}^\infty P_k|\langle\delta_x*\rho_n^{*k}-\delta_x*\rho^{*k},\phi\circ h_{k+1}\rangle| < 2\epsilon,
	\end{equation}

\noindent and thus, $|\Sigma_x(\rho_n) - \Sigma_x(\rho)| < 2\epsilon$.

\end{proof}

The only assertion that is left to be proven to justify the extension of $\Sigma_x$ and $\Sigma$ from $\mathbb P_p$ to $\mathbb P$ is that $\mathbb 
P_p$ is dense in $\mathbb P$. This theorem can be proven in the same way as Theorem~\ref{theorem-dense} but simpler. One can simply take $\psi_n = 
\rho*\omega_n$.

\subsection{Continuity of $A$ and the operators that generate $\sigma$ and $\sigma_x$ for a nonlinear selection function $f$ majorated by an 
exponential function}
\label{app-f-continuity}

The assumption of the Poisson distribution of the individuals in the compartments is essential here. We also assume that the phenotype-fitness relation 
is defined by a continuous selection function $f$. By its meaning, $f$ is expected to be nonnegative on the positive semiaxis. We, however, will treat 
a more general case, which will be useful for the question of an approximation of $f$. The phenotype is considered to be additive. The notations are 
the same as in Appendix~\ref{app-continuity}, except that by $\sigma_x^{f,\rho}$ and $\sigma^{f,\rho}$ we will denote the expressions (\ref{sigma-x}) 
and (\ref{sigma}), respectively, where $h_n \colon x \mapsto f(x)/n$, and we explicitly indicate the dependence on $f$ and $\rho$. By $\Sigma_x^f$ and 
$\Sigma^f$ we will denote the operators $\Sigma_x^f\colon \rho \mapsto \sigma_x^{f,\rho}$ and $\Sigma^f \colon \rho \mapsto \sigma^{f,\rho}$. We will 
also denote $N^f \deq \{\rho \in \mathbb P \,|\, \langle \sigma^{f,\rho},x\rangle = 0\}$, $N^{f,\mathcal I} \deq N^f \cap \mathbb P^\mathcal{I}$, and 
$N^{f,\mathcal I}_p \deq N^f \cap \mathbb P^\mathcal{I}_p$.

With a nonlinear selection function $f$, the situation becomes more complicated. In general, both $\sigma^{f,\rho}$ and $\sigma_x^{f,\rho}$ are 
not compactly supported densities anymore. Thus, the update operator $A^f \colon \rho \mapsto 
\dfrac{\langle\sigma_x^{f,\rho},y\rangle}{\langle\sigma^{f,\rho},y\rangle}\rho$ may not even be defined on all densities from $\mathbb P\setminus N^f$ 
or even from $\mathbb P_p^\mathcal{I}\setminus N^{f,\mathcal I}_p$.

We will prove first that the operator in question is, indeed, well defined for some class of functions $f$.

\begin{theorem}

Let $f$ be a continuous function. Let $a$ and $b$ be positive real numbers such that $|f(x)| \leqslant a\ch bx$ for all $x \in \mathbb R$. Then for any 
compactly supported probability density $\rho$ and for any $\lambda > 0$ the expectations of $\sigma_x^{f,\rho}$ and $\sigma^{f,\rho}$ are finite.

\label{theorem-finiteness}
\end{theorem}

\begin{proof}

To prove the theorem we will show that the series involved in $\langle \sigma_x^{f,\rho}, y \rangle$ and in $\langle \sigma^{f,\rho}, y \rangle$ 
converge absolutely, namely that
	\begin{equation}
	\sum_{n=0}^\infty\frac{e^{-\lambda}\lambda^n}{(n+1)!}|\langle \delta_x * \rho^{*n}, f(y)\rangle| < \infty,\quad
	\sum_{n=0}^\infty\frac{e^{-\lambda}\lambda^n}{(n+1)!}|\langle \rho^{*n+1}, f(y)\rangle| < \infty.
	\end{equation}

First, for any compactly supported probability density $\nu$ we have $|\langle \nu, f(x) \rangle| \leqslant \langle \nu, |f(x)|\rangle$. Then, using 
the estimate $|f(x)| \leqslant a\ch bx$, the expression for the moment generating function $\psi_\rho(s) \deq \langle \rho, e^{sx}\rangle$, and the 
fact that $\psi_{\delta_a}(s) = e^{as}$, we obtain the estimates
	\begin{multline} 
	\sum_{n=0}^\infty\frac{e^{-\lambda}\lambda^n}{(n+1)!}|\langle \delta_x * \rho^{*n}, f(y)\rangle| \leqslant 
	\sum_{n=0}^\infty\frac{e^{-\lambda}\lambda^n}{(n+1)!}\frac{a}{2}(e^{bx}\psi_\rho(b)^n+e^{-bx}\psi_\rho(-b)^n) \leqslant \\ \leqslant 
	\sum_{n=0}^\infty\frac{e^{-\lambda}\lambda^n}{(n+1)!}ae^{|bx|}\tilde\psi_\rho(b)^n = 
	ae^{|bx|-\lambda}\frac{e^{\lambda\tilde\psi_\rho(b)}-1}{\lambda\tilde\psi_\rho(b)} < \infty,
	\label{sigma-x-converge}
	\end{multline}

\noindent where $\tilde\psi_\rho(x) = \max\Big(\psi_\rho(x),\psi_\rho(-x)\Big)$, and
	\begin{multline}
	\sum_{n=0}^\infty\frac{e^{-\lambda}\lambda^n}{(n+1)!}|\langle \rho^{*n+1}, f(y)\rangle| \leqslant
	\sum_{n=0}^\infty\frac{e^{-\lambda}\lambda^n}{(n+1)!}\frac{a}{2}(\psi_\rho(b)^{(n+1)}+\psi_\rho(-b)^{(n+1)}) \leqslant \\
	\leqslant \sum_{n=0}^\infty\frac{e^{-\lambda}\lambda^n}{(n+1)!}a\tilde\psi_\rho(b)^{(n+1)} =
	a\frac{e^{\lambda \tilde\psi_\rho(b)} - 1}{\lambda e^\lambda} < \infty.
	\label{sigma-converge}
	\end{multline}

The last relations in the chains follow from the fact that for any compactly supported probability density $\rho$ its moment generating function 
$\psi_\rho$ is positive and finite for any value of the argument.
\end{proof}

A counterexample to the theorem's statement with the dropped condition $|f(x)| \leqslant a\ch bx$ is given by the function $e^{x^2}$ and the 
density $\delta_1$. Indeed, in this case the expressions for $\sigma_x$ and $\sigma$ coincide and we have
	\begin{equation}
	\langle\sigma^{e^{x^2},\delta_1}_1,y\rangle = \langle\sigma^{e^{x^2},\delta_1},y\rangle = 
	\sum_{n=0}^\infty \frac{e^{-\lambda}\lambda^n}{(n+1)!}e^{(n+1)^2}.
	\end{equation}
This series is divergent as its positive terms increase with the increase of $n$.

The statement of theorem can be extended to a wider class of functions $f$ than merely continuous (keeping the majorating condition). This requires a 
construction of the fully developed theory of integration for Radon measures, which is possible but is not of an interest in this work.

The next two theorem establish the continuity of $A^f$, $\Sigma_x^f$, and $\Sigma^f$.

\begin{theorem}

For any interval $\mathcal I$, under conditions of Theorem~\ref{theorem-finiteness}, the operator $A^f$ is continuous on the subset $\mathbb 
P^\mathcal{I}\setminus N^{f,\mathcal I}$.

\label{theorem-Af-continuous}
\end{theorem}

\begin{proof}

The proof essentially repeats the proof of Theorem~\ref{sigma-continuous}.

Let $\rho_n \in \mathbb P^\mathcal{I}$ be a sequence that approaches some $\rho \in \mathbb P^\mathcal{I}$ in $\mathbb P^\mathcal{I}$. We will prove 
that $\langle \sigma_x^{f,\rho_n},y\rangle$, as functions of $x$, converge to $\langle \sigma_x^{f,\rho},y\rangle$ uniformly on $\mathcal I$. The 
convergence $\langle \sigma^{f,\rho_n},y\rangle \to \langle \sigma^{f,\rho},y\rangle$ is proven analogously. The both facts will imply that, if $\rho 
\in \mathbb P^\mathcal{I}\setminus N^{f,\mathcal I}$, then
	\begin{equation}
	\frac{\langle \sigma_x^{f,\rho_n},y\rangle}{\langle \sigma^{f,\rho_n},y\rangle} \to 
	\frac{\langle\sigma_x^{f,\rho},y\rangle}{\langle \sigma^{f,\rho},y\rangle}
	\end{equation}

\noindent uniformly on $\mathcal I$, and thus, $A^f(\rho_n) \to A^f(\rho)$ in $\mathbb P^\mathcal{I}\setminus N^{f,\mathcal I}$.

As $\supp\rho_n \subset \mathcal I$ and $\supp\rho \subset \mathcal I$, we have the pointwise convergence $\psi_{\rho_n} \to \psi_\rho$. Indeed, for 
any $t \in \mathbb R$ we have
	\begin{equation}
	\psi_{\rho_n}(t) - \psi_\rho(t) = \langle\rho_n - \rho,e^{xt}\rangle = \langle\rho_n - \rho,\eta_t(x)\rangle \to 0,
	\end{equation}

\noindent where $\eta_t \in C_{c+}$ such that $\forall x \in \mathcal I$ $\eta_t(x) = e^{xt}$.

Therefore, for any $n$, we have the following estimate
	\begin{equation}
	|\langle \delta_x*(\rho_n^{*k} - \rho^{*k}),f(y)\rangle| \leqslant 
	ae^{b\mathcal{L}}\left(\left(\sup_m \tilde\psi_{\rho_m}(b)\right)^k + \tilde\psi_\rho(b)^k\right) \leqslant
	ae^{b\mathcal{L}}\left(\sup_m \tilde \psi_{\rho_m}(b) + \tilde\psi_\rho(b)\right)^k,
	\end{equation}

\noindent where the notation is the same is in the proof of Theorem~\ref{theorem-finiteness}.

Let us denote $c = \sup\limits_m \tilde\psi_{\rho_m}(b) + \tilde\psi_\rho(b)$. It follows that for any $n$
	\begin{equation}
	\sum_{k=0}^\infty\frac{e^{-\lambda}\lambda^k}{(k+1)!}|\langle \delta_x * (\rho_n^{*k} - \rho^{*k}), f(y)\rangle| \leqslant
	\sum_{k=0}^\infty\frac{ae^{b\mathcal{L}-\lambda}(c\lambda)^k}{(k+1)!} =
	\frac{ae^{b\mathcal{L}-\lambda}\left(e^{c\lambda} - 1\right)}{c\lambda} < \infty,
	\end{equation}

\noindent and therefore, for any $\epsilon > 0$ there is $k_0$ such that $\sum\limits_{k = 
k_0}^\infty\frac{ae^{b\mathcal{L}-\lambda}(c\lambda)^k}{(k+1)!} < \epsilon$, and, thus, for any $n$
	\begin{equation}
	\sum_{k=k_0}^\infty\frac{e^{-\lambda}\lambda^k}{(k+1)!}|\langle \delta_x * (\rho_n^{*k} - \rho^{*k}), f(y)\rangle| < \epsilon.
	\end{equation}

From the other hand, as $\rho_n \to \rho$, for any $\epsilon > 0$ and any $k_0$ there is $n_0$ such that for any $n > n_0$
	\begin{equation}
	\sum_{k=0}^{k_0 - 1}\frac{e^{-\lambda}\lambda^k}{(k+1)!}|\langle \delta_x * (\rho_n^{*k} - \rho^{*k}), f(y)\rangle| < \epsilon.
	\end{equation}

Therefore, by (\ref{lemma-logic}), it follows that for any $x$, $\langle \sigma_x^{f,\rho_n}, y\rangle \to \langle \sigma_x^{f,\rho}, y\rangle$.

For any $\rho \in \mathbb P$ with bounded support the function $\langle \delta_x*\rho,f(y)\rangle$ is continuous in $x$. Indeed, we have $\langle 
(\delta_{x'}-\delta_x)*\rho,f(y)\rangle = \langle\rho, f(y+x') - f(y+x)\rangle$ and $f(y + x') \to f(y+x)$ uniformly on $\supp \rho$ as $x' \to x$. 
Therefore, $\langle\sigma_x^{f,\rho_n},y\rangle$ and $\langle\sigma_x^{f,\rho},y\rangle$ are continuous in $x$ as absolutely convergent series of 
continuous functions. This, in turn, implies that $\langle \sigma_x^{f,\rho_n}, y\rangle \to \langle \sigma_x^{f,\rho}, y\rangle$ uniformly on 
$\mathcal I$ as functions of $x$.

\end{proof}

Note that $N^{f,\mathcal I}$ is nowhere dense in $\mathbb P^\mathcal{I}$ and $N^{f,\mathcal I}_p$ is nowhere dense in $\mathbb P^\mathcal{I}_p$. 
Indeed, this follows from the implication $\neg(\langle\sigma^{f,\rho_n},x\rangle \to \langle\sigma^{f,\rho},x\rangle)$ $\Rightarrow$ $\neg(\rho_n \to 
\rho)$ , which, in turn, is equivalent to the proven implication $\rho_n \to \rho$ $\Rightarrow$ $\langle\sigma^{f,\rho_n},x\rangle \to 
\langle\sigma^{f,\rho},x\rangle$. Note also that the statement of the theorem can be extended to $\mathbb P^K$, where $K$ is any compact set with 
nonempty interior. In the case when $\forall x > 0$, $f(x) > 0$, we have either $N^{f,\mathcal I} = \{\delta_0\}$ or $N^{f,\mathcal I} = \varnothing$, 
so all these subtleties become irrelevant for the extensions by continuity of $A^f$ from $\mathbb P^\mathcal{I}_p$ to $\mathbb P^\mathcal{I}$.

\begin{theorem}

Under conditions of Theorem~\ref{theorem-finiteness}, the operators $\Sigma_x^f$ and $\Sigma^f$ are continuous on $\mathbb P$.

\end{theorem}

\begin{proof}

Note that for any $\phi \in C_c$ we have $\sup\limits_x \left|\phi\left(\dfrac{f(x)}{n}\right)\right| \leqslant \sup\limits_x |\phi(x)|$. After that, 
the proof literally repeats the proof of Theorem~\ref{sigma-continuous}.

\end{proof}

\subsection{
Polynomial selection function and sums of exponential with polynomial coefficients}
\label{app-polynomial}

Let the total fitness in a compartment with $n$ individuals characterized by phenotypes $x_1$, \ldots, $x_n$ be given by $f(x_1 + \ldots + x_n)$, where
	\begin{equation}
	f(x) = a_0 + a_1 x + \ldots + a_m x^m.
	\end{equation}

As it is shown in Section~\ref{polynomial}, to find $\sigma_x$ and $\sigma$ it is enough to consider $(s^k_n)_x = \langle \delta_x*\rho^{*n}, y^k 
\rangle$ and $s^k_n = \langle \rho^{*n+1},y^k\rangle$ for any $k$, $0 \leqslant k \leqslant m$, and then to find the sums 
$\sum\limits_n\frac{P_n}{n+1}(s^k_n)_x$ and $\sum\limits_n\frac{P_n}{n+1}s^k_n$, $P_n = e_{}^{-\lambda}\lambda^n/n!$.

Note that the operation of the multiplication of a generalized function $\rho$ by the parameter ($x\rho$) is a derivation on a convolution algebra, 
that is it is linear and obeys the Leibniz rule:
	\begin{equation}
	x(\rho_1 * \rho_2) = (x\rho_1)*\rho_2 + \rho_1*(x\rho_2).
	\label{leibniz}
	\end{equation}

\noindent Indeed, for any $\phi \in C_c$ we have
	\begin{multline}
	\langle x(\rho_1*\rho_2), \phi\rangle = \langle \rho_1*\rho_2, x\phi\rangle = \langle \rho_1\otimes\rho_2,(x_1+x_2)\phi(x_1+x_2)\rangle=\\
	=\langle (x_1\rho_1)\otimes\rho_2 + \rho_1\otimes(x_2\rho_2),\phi(x_1+x_2)\rangle = \langle (x\rho_1)*\rho_2 + \rho_1*(x\rho_2),\phi\rangle.
	\end{multline}

As $\langle \rho,y^m\rangle = \overline{x^m}$ and $\langle \delta_x,y^m\rangle = x^m$, $(s^k_n)_x = \langle y^k(\delta_x*\rho^{*n}),1\rangle$ is a 
linear combination of all expressions of the form $x^\alpha\prod\limits_i\overline{x^{\beta_i}}^{\gamma_i}$, where all $\beta_i$ are different and 
$\alpha + \sum\limits_i(\beta_i + \gamma_i) = k$. For example, for $k=2$ we have $x^2$, $x \bar x$, $\bar x^2$, $\overline{x^2}$, for $k = 3$ we have 
$x^3$, $x^2 \bar x$, $x \bar x^2$, $x \overline{x^2}$, $\bar x^3$, $\bar x \overline{x^2}$, $\overline{x^3}$, for $k = 4$ we have $x^4$, $x^3\bar x$, 
$x^2\bar x^2$, $x^2 \overline{x^2}$, $x \bar x^3$, $x \bar x \overline{x^2}$, $x \overline{x^3}$, $\bar x^4$, $\bar x^2\overline{x^2}$, $\bar x 
\overline{x^3}$, $\overline{x^2}^2$, $\overline{x^4}$, etc. These expressions enter $(s^k_n)_x$ with coefficients that are (nonnegative) polynomials in 
$n$ of the form $n(n-1)\ldots(n-l+1)a_l(k)$ or just $a_0(k)$ ($a_i(k) \in \mathbb N$). The polynomial (in $k$) functions $a_i(k)$ can be in principle 
found using some combinatorics. At the very least, they are algorithmically computable.

The sums $\sum\limits_n P_n (s^k_n)_x/(n+1)$ can be evaluated using the identity
	\begin{eqnarray}
	\sum_{n=0}^\infty \frac{\lambda^n}{n!}\frac{n(n-1)(n-2)\ldots(n-p+1)}{n+1} = \lambda^p\sum_{n=0}^\infty 
	\frac{\lambda^n}{n!}\frac{1}{n+p+1} =
	\lambda^p \frac{d}{d\lambda^p}\frac{e_{}^\lambda - 1}{\lambda}.
	\end{eqnarray}

Likewise, $s^k_n$, being $\langle x^k \rho^{*n+1},1 \rangle$, is a linear combination of all expressions of the form 
$\prod\limits_i\overline{x^{\beta_i}}^{\gamma_i}$, where all $\beta_i$ are different and $\sum\limits_i(\beta_i + \gamma_i) = k$. The coefficients in 
front these expressions in $s^k_n$ are of the form $(n+1)n(n-1)\ldots(n-l+1)a_l(k)$, $l \in \mathbb N$. The summation of $\sum\limits_n P_n 
s^k_n/(n+1)$ is trivial.\\

Let us now consider that the total fitness in a compartment with $n$ individuals with phenotypes $x_1,\ldots,x_n$ is given by $f(x_1+\ldots+x_n)$, 
where $f$ is a linear combination of exponential functions with polynomial coefficients:
	\begin{equation}
	f(x) = \sum_{j=1}^m p_j(x) e_{}^{a_j x},
	\label{polynom-exp}
	\end{equation}

\noindent where $p_j$ are polynomials. It is again sufficient to find $\langle \delta_x*\rho^{*n}, y^k e_{}^{ay} \rangle$ and $\langle \rho^{*n+1}, 
y^k e_{}^{a y} \rangle$ for any $k$ and $a$. Note that $\langle x\rho, e_{}^{ax}\rangle = \left.\dfrac{d}{ds}\psi(s)\right|_{s=a}$, where $\psi(a) = 
\langle \rho, e_{}^{ax} \rangle$ (this is directly related to (\ref{leibniz})). Therefore, $\langle \delta_x*\rho^{*n},y^k e_{}^{ay}\rangle = 
\left.\dfrac{d^k}{ds^k}(e_{}^{sx}\psi(s)^n)\right|_{s=a}$ and $\langle \rho^{*n+1},y^k e_{}^{ay}\rangle = 
\left.\dfrac{d^k}{ds^k}\psi(s)^{n+1}\right|_{s=a}$. These expressions can be summed with $e_{}^{-\lambda}\lambda^n/(n+1)!$ in the same way as in the 
pure polynomial case. Thus, every function $f$ of the form (\ref{polynom-exp}) allows to write the update equation in closed form.

\subsection{Approximation of $f$ using truncated Fourier series}
\label{app-fourier}

We assume the Poisson distribution of the individuals in the compartments. We also assume that the phenotype-fitness relation is defined by a 
continuous selection function $f$. The activity is considered to be additive. By $\sigma_x^{f,\rho}$ and $\sigma^{f,\rho}$ we will denote the 
expressions (\ref{sigma-x}) and (\ref{sigma}), respectively, where $h_n \colon x \mapsto f(x)/n$, and we explicitly indicate the dependence on $f$ and 
$\rho$.

We will call a trigonometric polynomial of order $n$ and period $T$ any function from $\mathbb R$ to $\mathbb C$ of the form
	\begin{equation}
	p(x) = \sum_{k=-n}^{n}c_k e^{\frac{2\pi ikx}{T}},\quad c_k \in \mathbb C,\quad c_{-n}c_n \neq 0.
	\end{equation}

The trigonometric polynomial $p$ is called real if $p(\mathbb R) \subset \mathbb R$, where the second $\mathbb R$ is understood as the natural 
embedding in $\mathbb C$. $p$ is real if and only if $\forall k$ $c_{-k} = \bar c_k$.

Let $L(I)$ mean the length of the interval $I$. We call the truncated to order $n$ Fourier series of the function $f$ on the interval $I$ the 
trigonometric polynomial
	\begin{equation}
	f_n(x) = \sum_{k=-n}^n a_k e^{\frac{2\pi ikx}{L(I)}},\quad a_k = \frac{1}{L(I)}\intL{I}{} f(x)e^{-\frac{2\pi ikx}{L(I)}}dx.
	\end{equation}

We need the following known fact.

\begin{theorem}
For any $\epsilon > 0$ and any periodic continuous function 
$f$ with period $T$ there exists a trigonometric polynomial $p$ with period $T$ such that $\sup\limits_x |f(x) - p(x)| < \epsilon$. Furthermore, $p$ 
can be constructed from the Fourier series of $f$, namely if $f_n$ is the truncated to order $n$ Fourier series of $f$, then, for large enough $n$, 
one can take $p = (f_0 + f_1 + \ldots + f_n)/(n+1)$.
\label{theorem-fejer}
\end{theorem}

The proof can be found, for example, in \cite{Rudin1}. The last statement is known as Fej\'er's theorem. Note that all $p$ constructed in this way for 
a real function $f$ are real.

We start with an observation that for any probability density $\rho$ the following holds. If $\supp \rho \subset [-d,d]$, then $\supp \rho^{*n} 
\subset [-nd, nd]$ and $\supp \delta_x * \rho^{*n-1} \subset [-nd, nd]$ for any $x \in \supp \rho$. For any $\rho$ we will denote $I^\rho$ some 
interval $I^\rho = [-d,d]$ such that $\supp \rho \subset I^\rho$. For example, one can use the smallest interval with these properties. We will also 
introduce intervals $I^\rho_n = [-nd,nd]$ (not to be confused with $I_k$ used later) for the same $d$ that defines $I^\rho$.

We will prove the following main theorem.

\begin{theorem}
Let $f$ be a continuous function $f\colon \mathbb R \to \mathbb R$ exponentially bounded in the following sense: There are positive numbers $a$ and $b$ 
such that $|f(x)| \leqslant a\ch bx$ for all $x$. Then for any compactly supported probability density $\rho$ there exists a sequence $k \mapsto 
(I_k,p_k)$ of pairs of closed intervals $I_k$ and of real trigonometric polynomials $p_k$, where $p_k$ approximates $f$ on $I_k$, such that $\langle 
\sigma_x^{p_k,\rho},y\rangle \to \langle \sigma_x^{f,\rho},y\rangle$ homogeneously on $\supp \rho$ as a function of $x$, $\langle 
\sigma^{p_k,\rho},y\rangle \to \langle \sigma^{f,\rho},y\rangle$ as a sequence of numbers, while $\sigma_x^{p_k,\rho} \to \sigma_x^{f,\rho}$ for any $x 
\in \supp \rho$ and $\sigma^{p_k,\rho} \to \sigma^{f,\rho}$ in the sense of generalized functions. $p_k$ can be constructed using Fourier series 
approximations of $f$ on the appropriate intervals.
\label{f-approx}
\end{theorem}

As the logic of the proof is slightly convoluted, we will first formulate and prove two auxiliary lemmas.

\begin{lemma}

Under conditions of Theorem~\ref{f-approx}, for any $\epsilon > 0$ there exists $n_0 > 0$ such that
for any function $p$ that is bound by $\sup\limits_{x \in \mathbb R} |p(x)| < \sup\limits_{x \in I^\rho_{n_0}} |f(x)| + \epsilon$ the following holds:
	\begin{eqnarray}
	\forall x \in \supp \rho\quad
	&\displaystyle
	\left|\sum_{n=n_0}^\infty \frac{e^{-\lambda}\lambda^n}{(n+1)!}\langle\delta_x * \rho^{*n}, f - p \rangle\right|
	< \epsilon\quad \text{and}
	\label{first-statement1}\\
	&\displaystyle
	\left|\sum_{n=n_0}^\infty \frac{e^{-\lambda}\lambda^n}{(n+1)!}\langle\ \rho^{*n+1}, f - p \rangle\right|
	< \epsilon.
	\label{second-statement1}
	\end{eqnarray}

\label{lemma-more}
\end{lemma}

\begin{proof}

By the virtue of the estimate on $f$ from the conditions of Theorem~\ref{f-approx}, the relations of the statement of the lemma follow from the 
relations
	\begin{eqnarray}
	\forall x \in \supp \rho\quad
	&\displaystyle
	\sum_{n=n_0}^\infty \frac{e^{-\lambda}\lambda^n}{(n+1)!}\langle\delta_x * \rho^{*n}, |f| + \epsilon + a \ch\big(b L(I^\rho_{n_0})\big) \rangle
	< \epsilon\quad \text{and}
	\label{first-statement1-1}\\
	&\displaystyle
	\sum_{n=n_0}^\infty \frac{e^{-\lambda}\lambda^n}{(n+1)!}\langle\ \rho^{*n+1}, |f| + \epsilon + a \ch\big(b L(I^\rho_{n_0})\big) \rangle
	< \epsilon.
	\label{second-statement1-1}
	\end{eqnarray}

Let us prove (\ref{first-statement1-1}). Indeed, using the same reasoning as in the proof of Theorem~\ref{theorem-finiteness}, we have $\forall x \in 
\supp \rho$
	\begin{multline}
	\sum_{n=n_0}^\infty \frac{e^{-\lambda}\lambda^n}{(n+1)!}\langle\delta_x * \rho^{*n}, |f(y)| + a \ch\big(b L(I^\rho_{n_0})\big) + \epsilon\rangle
	\leqslant\\ \leqslant
	\sum_{n=n_0}^\infty \frac{e^{-\lambda}\lambda^n}{(n+1)!}\langle\delta_x * \rho^{*n}, a \ch\big(by) + a \ch\big(b L(I^\rho_{n_0})\big) + \epsilon 
	\rangle
	< \\ <
	\sum_{n=n_0}^\infty \frac{e^{-\lambda}\lambda^n}{(n+1)!} \left(ae^{b|x|}\tilde\psi_\rho(b)^n + a\left(e^{L(I^\rho)}\right)^{n_0} + \epsilon \right)
	\leqslant \\ \leqslant
	a\sum_{n=n_0}^\infty \frac{e^{-\lambda}\lambda^n}{(n+1)!} e^{bx_0}\tilde\psi_\rho(b)^n +
	a\sum_{n=n_0}^\infty \frac{e^{-\lambda}\lambda^n}{(n+1)!} \left(e^{L(I^\rho)}\right)^n +
	\epsilon\sum_{n=n_0}^\infty \frac{e^{-\lambda}\lambda^n}{(n+1)!},
	\label{n-condition}
	\end{multline}

\noindent where $\psi_\rho(t) = \langle \rho,e^{tx}\rangle$, $\tilde\psi_\rho(t) = \max\Big(\psi_\rho(t), \psi_\rho(-t)\Big)$, and $x_0 = \max(|\inf 
\supp \rho|, |\sup \supp \rho|)$.

The statement (\ref{first-statement1-1}) follows from the fact that the series in the last expression have positive terms and converge to finite 
numbers, if summed started from $n = 0$. Indeed, this means that for any $\epsilon > 0$ there is $n_0$ such that the last expression is smaller than 
$\epsilon$. It can be proven analogously that the same $n_0$ fulfills the statement~(\ref{second-statement1-1}).

\end{proof}

\begin{lemma}

Under conditions of Theorem~\ref{f-approx}, for any $\epsilon > 0$ and any $n_0 > 0$ there exists a real trigonometric polynomial $p$ such that 
$\sup\limits_{x \in I^\rho_{n_0}} |f(x) - p(x)| < \epsilon$, $\sup\limits_{x \in \mathbb R}|p(x)| < \sup\limits_{x \in I^\rho_{n_0}}|f(x)| + \epsilon$, 
as well as
	\begin{eqnarray}
	\forall x \in \supp \rho\quad
	&\displaystyle
	\left|\sum_{n=0}^{n_0-1} \frac{e^{-\lambda}\lambda^n}{(n+1)!} \langle \delta_x * \rho^{*n}, f - p \rangle\right| < \epsilon \quad \text{and}
	\label{first-statement2}
	\\
	&\displaystyle
	\left|\sum_{n=0}^{n_0-1} \frac{e^{-\lambda}\lambda^n}{(n+1)!} \langle \rho^{*n+1}, f - p \rangle\right| < \epsilon.
	\label{second-statement2}
	\end{eqnarray}

\label{lemma-less}
\end{lemma}

\begin{proof}

Let $d$ be the number such that $I^\rho = [-d,d]$. For any $n \leqslant n_0-1$, we have $\supp \delta_x * \rho^{*n} \subset I^\rho_{n_0}$ and $\supp 
\rho^{*n+1} \subset I^\rho_{n_0}$. Let us extend $f|_{I^\rho_{n_0}}$ to $I^\rho_{n_0+1}$ by a function $\hat f$ such that $\hat f(x) = \alpha x + 
\beta_1$ for any $x \in [-(n_0+1)d, -n_0 df]$ and $\hat f(x) = \alpha x + \beta_2$ for any $x \in [n_0 d, (n_0+1)d]$, where $\alpha$, $\beta_1$, and 
$\beta_2$ are selected to fulfill

	\begin{equation}
	\hat f\big(-(n_0+1)d\big) = \hat f\big((n_0+1) d\big) = \frac{f(-n_0 d) + f(n_0 d)}{2},\quad
	\hat f(-n_0 d) = f(-n_0 d),\quad \hat f(n_0 d) = f(n_0 d).
	\end{equation}

\noindent Function $\hat f$ is continuous on $I^\rho_{n_0+1}$ and can be extended to the whole $\mathbb R$ as a periodic continuous function $\check f$ 
with period $L(I^\rho_{n_0+1})$. By Theorem~\ref{theorem-fejer}, using the truncations of the Fourier series for $\check f$ up to some order, we can 
construct a real trigonometric polynomial $p$ such that for any $\epsilon > 0$\, $\sup\limits_{x \in \mathbb R}
|\check f(x) - p(x)| < \epsilon$. It follows that for any $\epsilon > 0$ we can find $p$ such that $\sup\limits_{x \in I^\rho_{n_0}} |f(x) - p(x)| <
\epsilon$. Furthermore, by construction, $\sup\limits_{x \in \mathbb R}|p(x)| < \sup\limits_{x \in I^\rho_{n_0}}|f(x)| + \epsilon$.

Let us consider the right-hand side of the statement (\ref{first-statement2}). For any $x \in \supp \rho$ we have

	\begin{equation}
	\sum_{n=0}^{n_0-1} \frac{e^{-\lambda}\lambda^n}{(n+1)!} \langle \delta_x * \rho^{*n}, |f - p| \rangle <
	\epsilon\sum_{n=0}^{n_0-1} \frac{e^{-\lambda}\lambda^n}{(n+1)!} < \epsilon g(\lambda) \leqslant \epsilon.
	\label{delta-x-relation}
	\end{equation}

\noindent Here we essentially used the inclusion $\supp \delta_x * \rho^{*n} \subset I^\rho_{n_0}$ for any $x \in \supp \rho$. The relation 
(\ref{delta-x-relation}) implies the relation (\ref{first-statement2}).

The statement (\ref{second-statement2}) can be proven analogously, using the fact that $\supp \rho^{*n+1} \subset I^\rho_{n_0}$ for any $n < 
n_0$.

\end{proof}

\begin{proof}[Proof of Theorem \ref{f-approx}]

Let us symbolically rewrite the statement of Lemma~\ref{lemma-more} as $\forall \epsilon > 0\,\exists n_0 > 0\,\forall p\colon A(\epsilon, n_0, p)$. In 
the same manner, let us symbolically rewrite the statement of Lemma~\ref{lemma-less} as $\forall \epsilon > 0\,\forall n_0 > 0\,\exists p\colon 
B(\epsilon, n_0, p)$. Then (\ref{lemma-logic}) implies $\forall \epsilon >0\, \exists n_0 > 0\, \exists p\colon A(\epsilon, n_0, p) \wedge B(\epsilon, 
n_0, p)$. As the statement $A$ estimates the sums from 0 to $n_0 - 1$, and the statement $B$ estimates the sums from $n_0$ to $\infty$, the statement 
$A \wedge B$ gives estimates on the sums from 0 to $\infty$. Therefore, the consequence of the above is explicitly read as following. For any positive 
$\epsilon$ there exists a number $n_0$ and a trigonometric polynomial $p$ (which can be constructed based on an approximation of $f$ by truncated 
Fourier series, as in the proof of Lemma~\ref{lemma-less}) such that $\sup\limits_{x \in I^\rho_{n_0}} |f(x) - p(x)| < \epsilon$ and

	\begin{eqnarray}
	\forall x \in \supp \rho\quad
	&\displaystyle
	\left|\sum_{n=0}^\infty \frac{e^{-\lambda}\lambda^n}{(n+1)!} \langle \delta_x * \rho^{*n}, f - p \rangle\right| < \epsilon \quad \text{and}
	\\
	&\displaystyle
	\left|\sum_{n=0}^\infty \frac{e^{-\lambda}\lambda^n}{(n+1)!} \langle \rho^{*n+1}, f - p \rangle\right| < \epsilon.
	\end{eqnarray}

The part of the theorem with $\langle \sigma_x^{p_k,\rho},y\rangle \to \langle \sigma_x^{f,\rho},y\rangle$ and $\langle \sigma^{p_k,\rho},y\rangle \to 
\langle \sigma^{f,\rho},y\rangle$ is proven, if we take, for example, $\epsilon_k = 1/(1 + k)$, $k \in \mathbb N$, and if we take as $I_k$ the interval 
$I^\rho_{n_0}$ and as $p_k$ the trigonometric polynomial $p$, where $n_0$ and $p$ are provided by the above statement for $\epsilon = \epsilon_k$.

Let us choose any $\phi \in C_c$. Let us denote $\Phi \deq \sup\limits_x |\phi(x)|$ and $s_k \deq \sup\limits_{x \in I_k} 
|\phi\left(p_k(x)\right) - \phi\left(f(x)\right)|$. As $\phi$ is continuous and compactly supported it is also uniformly
continuous on $\mathbb R$. Therefore, we have $s_k \to 0$, as $k \to \infty$. Furthermore, for any $n$ we have
	\begin{equation}
	\sup_{x \in I_k} \left|\phi\left(\frac{p_k(x)}{n+1}\right)- \phi\left(\frac{f(x)}{n+1}\right)\right| \leqslant s_k
	\end{equation}
and
	\begin{equation}
	\sup_x \left|\phi\left(\frac{p_k(x)}{n+1}\right)- 
	\phi\left(\frac{f(x)}{n+1}\right)\right| \leqslant 2\Phi.
	\end{equation}

Using the same technique of splitting the series into two parts by $n_0(k) = n_0(\epsilon_k)$ and using the former estimate for the initial part of the 
sum of the series and the latter estimate for the rest of the series by the same logic and taking into account that $n_0(\epsilon)$ provided by the 
Lemma~\ref{lemma-more} ever grows with the decay of $\epsilon$, we conclude that
	\begin{equation}
	|\langle \sigma_x^{p_k,\rho} - \sigma_x^{f,\rho},\phi\rangle| \leqslant
	s_k + 2\Phi e^{-\lambda}e_{n_0(k)}^\infty(\lambda) \to 0,\quad k \to \infty,
	\end{equation}

\noindent where $\displaystyle e_m^\infty(x) \deq \sum_{j = m}^\infty \frac{x^j}{j!}$.

In the same way we prove $\sigma^{p_k,\rho} \to \sigma^{f,\rho}$.

\end{proof}

This theorem implies that for any compactly supported $\rho$ with $\langle \sigma^{f,\rho},x\rangle \neq 0$ we have the convergence $A^{p_k}(\rho) \to 
A^f(\rho)$, as $k \to \infty$. Note that the sequence $p_k$ depends on $\rho$. This is due to the involvement of $\psi_\rho$ in (\ref{n-condition}). 
However, if $f$ is a bounded function, then this dependence can be dropped from (\ref{n-condition}), and the choice of $n_0$, and thus of $p_k$, 
becomes independent of the current distribution. In this case we can state the pointwise convergence $A^{p_k} \to A^f$ (on $\mathbb P^K$ for some 
compact $K$). The proof of the theorem is constructive. However, it is not optimized for applications.}

%\bibliographystyle{ieeetr}
%\bibliography{NaturalSelectionInDroplets}

\end{document}